\documentclass{amsart}

\usepackage{stmaryrd, mathrsfs, enumerate, hyperref}

\title[Rosen universes]{Sachs equations and plane waves, I:\\ Rosen universes}
\author{Jonathan Holland}
\address{Compunetix\\
  2420 Mosside Blvd \# 1\\
  Monroeville, PA 15146}
\author{George Sparling}
\address{University of Pittsburgh\\
  Department of Mathematics\\
  301 Thackeray Hall\\ Pittsburgh, PA 15260
}
\newtheorem{theorem}{Theorem}
\newtheorem{lemma}{Lemma}
\newtheorem{corollary}{Corollary}
\newtheorem{definition}{Definition}

\newcommand{\op}{\operatorname}
\newcommand{\tr}{\operatorname{tr}}

\begin{document}


\begin{abstract}
This article, the first in a series,  analyzes the general theory of plane wave spacetimes.  Following Dmitri Aleekseevsky, these are defined as spacetimes admitting a group of dilations leaving invariant a smooth curve. If this curve is specified as part of the structure, the spacetime is termed a Penrose limit, whose theory was developed first by Roger Penrose.  The main result is that every plane wave is a Rosen universe, a generalization of the smooth metrics of Albert Einstein and Nathan Rosen, allowing for certain isolated co-ordinate singularities; the latter are characterized.  We conclude with an extended example, using the techniques developed in the article to associate a vacuum plane wave in four dimensions to any hyperbolic billiard trajectory.
\end{abstract}
\maketitle

\tableofcontents

\section{Introduction}
Since the famous address of Hermann Minkowski \cite{minkowski}, it has been apparent that understanding the light cones of spacetime is key.  In the language of Minkowski, Marcel Grossmann and Albert Einstein, spacetime is a four-dimensional manifold equipped with a Lorentzian metric, such that the light cones are formed by the null geodesics of that metric.  In principle each such null geodesic zips through the entirety of spacetime in an instant, as does any particle that moves along the geodesic, such as (it is believed) a photon, a graviton, or,  perhaps, a neutrino.  

One aim of the twistor theory  of Roy Kerr \cite{kerr1963gravitational}, Ezra Newman, Roger Penrose and Rainer Sachs (a general reference is \cite{penrose1984spinors}, \cite{penrose1986spinors}) is to analyze spacetime in terms of its space of null geodesics: they are regarded  as primary entities, thinking of the points of spacetime as being derived.  For example, in dimension four,  it is standard that the null geodesics of conformally compactified Minkowski spacetime naturally form a five-dimensional Cauchy-Riemann manifold with indefinite Levi-form, called the null projective twistor space,  the (compact) flat model in the sense of Shiing-Shen Chern and J\"{u}rgen Moser,  such that each point of spacetime is represented by an holomorphic Riemann sphere lying inside the Cauchy--Riemann manifold.   Then a fundamental problem is to understand the twistor spaces of general spacetimes (so the space of unparametrized null geodesics) and their relation to the points of spacetime; this problem has dominated the thoughts of the second author for some sixty years, since being introduced to the area by Hermann Bondi.

In a spacetime of dimension $n + 2 \ge 3$ the space of null geodesics is naturally a contact manifold of dimension $2n + 1$, so possesses a distinguished $2n$-dimensional  subspace of the tangent space at each point, such that this space possesses a symplectic structure defined up to scale.   In particular the space of Lagrangian subspaces (so of maximal dimension $n$, annihilating the symplectic form) is canonical, of dimension $n(n + 1)$ at each point.   Lifting up to the spacetime gives a $2n$-dimensional collection of null geodesics infinitesimally near a given null geodesic, called abreast and within that collection its family of $n$-dimensional Lagrangian subspaces.   The evolution of these subspaces as one moves along a given null geodesic is described by the  Sachs equations, a first-order conformally invariant non-linear system.  In turn the properties of the Sachs equations form the subject of the present series. 

Penrose \cite{penrose1976any} brilliantly showed that one could encode the information of the neighborhood of a given null geodesic in terms of a simple kind of spacetime, a plane wave.   These are now called Penrose limits.  So in for a four dimensional spacetime, there is a five parameter set of such Penrose limits, one for each null geodesic.  We think of these as {\em osculating} the given spacetime to second order around each given null geodesic.   In turn the Penrose limits encode the Sachs equations.

Here we take the Sachs equations and the Penrose limits as fundamental building blocks.   For us it turns out to be important to develop the theory in arbitrary dimension.  In particular, this permits us to contemplate the theory in infinite dimension, which as we show, will allow us to perform a kind of Fourier analysis for the spacetime.   Thus there are three levels of theory: finite dimensions, which allows only a finite number of ``Fourier'' components, countable dimension, which is appropriate for periodic or almost periodic spacetimes and uncountable dimension, appropriate for completely general spacetimes.  The main idea here is to express the information of the Penrose limits in terms of homogeneous plane waves in higher dimensions; these we term {\em microcosms}.   We think the essence of the theory has crystallized and is ready for presentation.

In the case of conformally flat four-dimensional spacetime, complex numbers are critical in the twistor theory, the spacetime points forming Riemann spheres, as mentioned briefly above.  A priori it is not clear that complex  numbers are needed for conformally curved spacetimes, or in higher dimension.   For microcosms,  however, {\em provided they have non-negative energy}, it is a striking fact that complex numbers are essential to their understanding. Thus we are claiming a ``meta-theorem''  that complex numbers play a vital role for general physical spacetimes,  as has long been conjectured to be the case in the twistor community.

In the present work, we cover the general theory of plane waves, placing the Sachs equation at the fore, proving all results rigorously and from scratch.  Excellent general prior references are   \cite{blau2011plane} and \cite{blau2003homogeneous}, which especially emphasize the importance of plane wave spacetimes in the context of string theory.  An innovation is the definition of a family of spacetimes that we call Rosen universes and a theorem  characterizing this family.   Three recent works emphasize the current value of this study, first,  work of the authors,  Aristotelis Panagiotopoulos and Marios Christodoulou \cite{panagiotopoulos2023incompleteness}, using plane waves to show that spacetime theory has a certain minimal level of complexity, related to the shifts of Bernouilli, second the wonderful work of Maciej Dunajski and Penrose \cite{DunajskiPenrose}, which employs plane wave spacetimes in five dimensions to analyze quantum Newton-Hooke theory, and third a remarkable discussion of Maxwell theory due to   Carlos García-Meca, Andrés Macho Ortiz, and Roberto Llorente Sáez \cite{garcia2020supersymmetry} wherein supersymmetry is encoded as time reversal.

When properly understood, the Sachs equation is the ordinary differential equation
\begin{equation}\label{SachsIntro}
  S'(u) + S(u)^2 + p(u) = 0
\end{equation}
where $S$ is an unknown $n\times n$ real matrix, depending differentiably on a real $u$, and $p(u)$ is a given real symmetric $n\times n$ matrix, depending smoothly on $u$.  In the applications of interest, $S$ is the expansion tensor of a congruence of null geodesics.  Our signs are such that $p>0$ corresponds to the case of positive energy in spacetime.  A solution of the Sachs equation governs the focusing of a congruence of null geodesics, and is the key ingredient in the (optical) Raychaudhuri \cite{raychaudhuri1955relativistic} equation in four dimensions.  The skew part $2\omega = S-S^T$ is interpreted as the vorticity of a congruence of null geodesics, the trace $\theta = \tr S$ is the expansion of the congruence, and $2\sigma = S + S^T - 2n^{-1}\theta$ is the shear.  The Raychaudhuri equation is the trace of \eqref{SachsIntro}:
$$\theta' + n^{-1}\theta^2 + \tr(\sigma^2) + \tr(\omega^2) + \tr p = 0.$$
Since the congruences in the present work are non-rotating, the vorticity vanishes $(\omega=0)$, and the expansion tensor $S$ is symmetric.

Even in the scalar case, when $n=1$, some of the general character of \eqref{SachsIntro} can be intuited.  Suppose that $n=1$ and $p$ is constant.  The three main cases of constant $p$ are:
\begin{itemize}
\item $p=-1$: $\dot S + S^2-1=0$, $S=\tanh (u+c)$, 
\item $p=0$: $\dot S + S^2=0$, $S=(u+c)^{-1}$, and
\item $p=1$: $\dot S + S^2+1=0$, $S=\cot (u+c)$.  
\end{itemize}
Thus in vacuum $(p=0)$, the solution blows up at one endpoint of the domain.  When the energy is positive, the solution blows up at both endpoints of the domain.  When the energy is negative, the solution exists for all $u$.  A point where $S$ blows up is associated to the formation of caustics in the geodesic congruence, and positivity of the energy leads to the formation of caustics.  The focusing effects of congruences in plane waves appears to have been first appreciated in \cite{baldwin1926relativity}.

This paper proves two results that we claim as new.  First, we give a characterization of Rosen universes (Theorem \ref{AlekseevskyTheoremNew} and Theorem \ref{RosenUniverse}).  Second, we give a construction of a vacuum spacetime associated to any hyperbolic trajectory composed of geodesic arcs and horocycles, which we term a {\em hyperbolic drawing} (Theorem \ref{HyperbolicDrawingTheorem}). Also, we give a new proof of the fundamental theorem of Alekseevsky \cite{Alekseevsky} that also provides insight into the universality of Penrose limits \cite{blau2004universality}.

\section{Plane waves}\label{PlaneWavesSection}

We work with smooth real manifolds, which are paracompact, usually connected and simply connected, and usually of dimension at least three.  
\begin{definition}
  A {\em spacetime} is a pair $(\mathbb{M}, G)$, with $\mathbb{M}$ a simply connected smooth manifold and $G$ a smooth Lorentzian metric on $\mathbb{M}$.
\end{definition}
We write the dimension of the manifold $\mathbb{M}$ as $n + 2$, where $n$ is a non-negative integer and we take the metric $G$ to have signature $(1, n + 1)$. 
\begin{definition}
  A {\em dilation group} for a spacetime $(\mathbb{M}, G)$ is a one-parameter group,  $\mathcal{D}_t$, of diffeomorphisms, acting properly on $\mathbb{M}$, such that $\mathcal D_t^*G=e^{2t}G$. 
\end{definition}
A dilation group is thus generated by a smooth vector field $D$ such that $\mathcal{L}_D G = 2G$. Given such a dilation group, $\mathcal{D}$, denote by $\gamma_{\mathcal{D}}$ the fixed point set of $\mathcal{D}$, so the set of all points of $\mathbb{M}$ where the vector field $D$ vanishes. A dilation group is called {\em complete} if the limit $\lim_{t\to\infty}\mathcal D_t^{-1}x$ exists for all $x\in M$.  Each such limit is a fixed point of the dilation group.
\begin{definition}
    A spacetime $(\mathbb{M}, G)$ is called a {\em plane wave} if it is acted upon by some complete dilation group $\mathcal{D}$, such that its fixed point set $\gamma_{\mathcal{D}}$ is a smooth curve.
  \end{definition}
  Note that neither the particular group $\mathcal D$,  nor the fixed point set, is regarded as part of the structure.  Evary plane wave admits at least a $(2n+1)$-dimensional family of compatible dilations (Corollary \ref{ManyDilations}).  When we wish to regard the dilation as part of the structure, we make the definition:
\begin{definition}
    A {\em Penrose limit} is a triple $(\mathbb{M}, G, \mathcal{D})$ having the structure of a plane wave with (specified) dilation group $\mathcal D$.
      \end{definition}
 \begin{definition}  An {\em isomorphism of Penrose limits} is a smooth isometry that preserves $\mathcal D$.
  \end{definition}
  The extra structure of a particular dilation is sometimes of use to analyze the relevant parts of the symmetry group, by taming various degenerate kinds of spaces (such as flat space), that can arise.
\begin{lemma}
The curve fixed by a dilation group of a plane wave spacetime is a null geodesic.
\end{lemma}
\begin{proof}  First, the fixed curve is null because the norm of its tangent vector must be fixed by the dilation and hence is zero.  It is a geodesic because its acceleration vector must be both null, and orthogonal to the (null) tangent vector.  Because of the Lorentzian signature of the metric, this implies that the acceleration is proportional to the tangent vector.
\end{proof}
 \begin{definition}
  The   {\em wave fronts} for a Penrose limit spacetime are the unstable manifolds of the dilations.
  \end{definition}
By ``unstable manifold'' we here mean the following special case of the general concept in dynamics \cite{Ruelle}.  For the dilations considered here, the following definition suffices: an unstable manifold is a maximal invariant manifold of the group that contains a single fixed point.  For the dilation group, the fixed point lies on the central null geodesic.  The manifold is ``unstable'' in the sense that the dilation expands from the fixed point instead of contracting.  This definition shall be clarified by examples presently.

\subsection{The variety $\mathcal{A}$ of plane wave metrics}
\begin{definition} An {\em Euclidean vector space} is a  pair $(\mathbb{X}, T)$, consisting of a real vector space $\mathbb{X}$, of finite positive dimension, with dual space $\mathbb{X}^*$ and a linear isomorphism $T: \mathbb{X} \rightarrow \mathbb{X}^*$, such that 
\begin{itemize}\item $T(x)(y) = T(y)(x)$, for every $x$ and $y$ in $\mathbb{X}$, 
\item $T(x)(x) > 0$ for any non-zero $x \in \mathbb{X}$.
\end{itemize}
\end{definition}
We extend $T$ to an (anti)-automorphism of the tensor algebra of $\mathbb{X}$ and $\mathbb{X}^*$,  such that $T^2$ is the identity (here we have $T(\alpha \otimes \beta) = T(\beta) \otimes T(\alpha)$, for any tensors $\alpha$ and $\beta$).  For a tensor $\alpha$ of the tensor algebra, we abbreviate $T(\alpha)$ by $\alpha^T$ and we call $\alpha^T$ the transpose of $\alpha$.  
\begin{definition} An endomorphism $P$ of an Euclidean vector space $(\mathbb{X}, T)$ is  said to be {\em symmetric} if and only if $P = P^T$.  An endomorphism $P$ of an Euclidean vector space $(\mathbb{X}, T)$ is  said to be {\em skew} if and only if $P = - P^T$.
\end{definition} 
Each endomorphism $P$ of an Euclidean vector space $(\mathbb{X}, T)$  decomposes uniquely as $P = P_\odot + P_{\wedge}$, where the endomorphism  $P_\odot  = 2^{-1}(P + P^T)$ is symmetric and  the endomorphism $P_{\wedge} = 2^{-1}(P - P^T)$ is skew.

An umbrella for the plane wave metrics is denoted $\mathcal{A}$,  a collection of Lorentzian metrics,  which we now describe.  They constitute a slight generalization of the metrics of Dmitri Alekseevsky.

The spacetime manifold $\mathbb{M}$ is a product,  $\mathbb{M} = \mathbb{U} \times \mathbb{R} \times \mathbb{X}$, with $(\mathbb{X}, T)$ a fixed Euclidean real vector space and $\mathbb{U}$ an open real interval.

A point $X$ of $\mathbb{M}$ is denoted $X \leftrightarrow (u, v, x)$, with $u \in \mathbb{U}$, $v \in \mathbb{R}$ and $x \in \mathbb{X}$.  The co-ordinate differential is then $dX \leftrightarrow (du, dv, dx)$.    Dually, we have the co-ordinate derivative operators $\partial_X \leftrightarrow (\partial_u, \partial_v, \partial_x)$. So $dx$ takes values in the vector space $\mathbb{X}$, whereas $\partial_x$ takes values in the dual vector space $\mathbb{X}^*$.  

In the following the metric is described in terms of functions on $\mathbb{U}$,  taking values either in the reals or in the endomorphisms of $\mathbb{X}$.   Henceforth it will be understood that all such functions are globally defined on $\mathbb{U}$ and smooth, unless explicitly stated otherwise.
\begin{definition} The metrics of the class $\mathcal A$ are the tensors $G_\alpha(h)$ for $\mathbb{M} = \mathbb{U}\times \mathbb{R} \times \mathbb{X}$ defined by
  \begin{align}
 G_{\alpha}(h) &= \alpha du - dx^T h(u) dx, \\
    \notag \alpha &=  2dv - b(u)v du + x^T p(u) x du - x^T q(u) dx.
  \end{align}
Each of the quantities $h, b, p$ and $q$ is smoothly defined on $\mathbb{U}$.  Also,  for each $u\in \mathbb{U}$, $b(u)$ is real-valued, whereas each of $h(u)$, $p(u)$ and $q(u)$ is an endomorphism of $\mathbb{X}$, with $h(u)$  symmetric and positive definite.  
\end{definition}  
\begin{definition} The {\em standard dilation group} for a metric $G_{\alpha}(h)$ of the collection $\mathcal{A}$, which renders the spacetime a plane wave, is the group of diffeomorphisms of $\mathbb{M}$ given by the formula, valid for any real $t$ and any $X = (u,  v, x)\in \mathbb{M}$:
\[ \mathcal{D}_t(u, v, x) \rightarrow  \left(u, e^{2t}v, e^t x\right).\]
The generating vector field is $D = 2v\partial_v + x\partial_x$. Then the null geodesic $\gamma_{\mathcal{D}}$ is the curve $v = 0, x = 0$, with $u \in \mathbb{U}$ arbitrary.
\end{definition}
For the study of Penrose limits, we will always use the standard dilation group and we call the curve $\gamma_{\mathcal{D}}$ the central null geodesic.  The wave fronts, which are the unstable manifolds of $\mathcal D$, are the $u=$constant hypersurfaces: the group $\mathcal D$ acts tangentially to the wave fronts, and they are the maximal invariant manifolds for $\mathcal D$ that contain a single fixed point (the point $\gamma(u)$).

For later convenience, in the formula for $G_{\alpha}(h)$ just given, we do not impose any symmetry conditions on the endomorphisms $p$ and $q$, even though after an appropriate diffeomorphism, we may assume without loss of generality that $p$ is symmetric and $q$ skew, as we now show.  In the following discussion,  all diffeomorphisms will preserve the standard dilation group. 
\begin{lemma}Up to an isometry preserving the dilations and the central null geodesic, we may assume that the endomorphisms $p$ and $q$ of the metric $G_\alpha(h)$ are respectively symmetric and skew.
\end{lemma}
\begin{proof}
Suppressing the $u$-dependence,  we can re-express the one-form $\alpha$ as follows:
\begin{align*}
 \alpha &=  2dv - bv du + x^T p x du - x^T q dx \\
 &= 2dv - bv du + x^T p_\odot x du - x^T q_\wedge dx -  2^{-1} d(x^T q_\odot x) + 2^{-1} x^T q'_\odot x du.\\
\end{align*}
Here and in the following a prime denotes the derivative with respect to $u$.   Now put $V = v - 4^{-1}  x^T q'_\odot x$.  Notice that this formula is compatible with the dilation action, as required.  Then we have:
\begin{align*}
 \alpha &=  2dV - b\left(V + 4^{-1}  x^T q'_\odot x\right)  du + x^T p_\odot x du - x^T q_\wedge dx  + 2^{-1} x^T q'_\odot x du, \\
 \alpha &= 2dV - bV du + x^TP x du - x^T Q dx, \\
\end{align*}
  \[Q = q_{\wedge}, \quad P = p_\odot - 4^{-1}b  q'_\odot.\]
Note that $Q$ is skew, whereas $P$ is symmetric and the lemma is proved.
\end{proof}
Dropping capital letters, we now have a standard form for a metric of type $\mathcal{A}$:
\begin{align*}
  G_{\alpha}(h) &= \alpha du - dx^T h(u) dx, \\
   \alpha &=  2dv - b(u)v du + x^T p(u) x du - x^T q(u) dx,
\end{align*}
\[p = p^T, \quad q = - q^T.\]
We define three sub-collections of the collection $\mathcal{A}$.
\begin{definition} The {\em Alekseevsky metrics},  giving a sub-collection $\mathcal{A}_\alpha$, are the special case of the collection $\mathcal{A}$, with $h(u) = I$, for each $u\in \mathbb{U}$, where $I$ denotes the identity endomorphism of $\mathbb{X}$:
  \begin{align}\label{Alekseevsky1} 
    G_{\alpha} &= \alpha du - dx^T dx, \\
    \notag \alpha &=  2dv - b(u)v du + x^T p(u) x du - x^T q(u) dx.
  \end{align}
Here $p$ is symmetric and $q$ is skew.
\end{definition}
\begin{definition} The {\em Brinkmann metrics},  giving a sub-collection $\mathcal{A}_\beta$.  These have $b = 0$, $q = 0$ and $h$ the identity:
  \begin{equation}\label{Brinkmann1}
    G_\beta(p)  = 2\,du\,dv + x^Tp(u) x \,du^2 - dx^T dx.
  \end{equation}
Here  $p(u)$ is symmetric and $p$ is smooth on $\mathbb{U}$.
\end{definition}
\begin{definition}   The {\em regular Rosen metrics},  giving a sub-collection $\mathcal{A}_\rho$.  These have $b = 0$, $p = 0$ and $q = 0$, so the metric reads:
  \begin{equation}\label{Rosen1}
    G_\rho(h)  = 2\,du\,dv  - dx^T h(u) dx.
  \end{equation}
These metrics were first introduced in by Albert Einstein and Nathan Rosen \cite{einstein1937gravitational}. Here, for each $u \in \mathbb{U}$, the endomorphism $h(u)$ is symmetric and positive definite and $h$ is smooth on $\mathbb{U}$.
\end{definition}

\begin{lemma}
  Each metric $G_\alpha(h), G_\alpha, G_\beta(p), G_\rho(h)$, equipped with the standard dilation group is a plane wave.
\end{lemma}
This lemma follows immediately from the definitions.

Lemma \ref{EveryBrinkmannRosen} below shows that, given a Brinkmann metric $G_\beta(p)$, there is a covering of $\mathbb U$ by open intervals $\mathbb U_i$, Rosen metric $G_\rho(h_i)$ on $\mathbb M_i=\mathbb U_i\times\mathbb R\times\mathbb X$, and smooth functions $\phi_i:\mathbb M_i\to\mathbb M$ preserving the family of wave fronts that are isometries onto their image of $G_\rho(h_i)$ to $G_\beta(p)$.  The local nature of this construction can be made global if we allow the symmetric form $h_i$ to develop isolated singularities.

\begin{definition}
  A {\em singular Rosen metric} is a symmetric tensor of the form \eqref{Rosen1}, where the smooth symmetric tensor $h(u)$ is positive semidefinite, and the set of singular points $\mathbb S\subset\mathbb U$ has no limit point in $\mathbb U$.
\end{definition}
Note that the singular set $\mathbb{S}$ here, the closed subset of $\mathbb{U}$ where $\det(h(u))$ vanishes, may be empty, finite or countably infinite. The singular Rosen plane wave $\mathbb M$ is equipped with the standard group of dilations, $\mathcal{D}$, as for the regular case.  Note that if all points of $\mathbb{M}$ with $u \in \mathbb{S}$ are deleted, this splits up the manifold $\mathbb M$ into a countable collection of {\em regular} Rosen plane waves, which we denote by $\mathbb M-\mathbb S := (\mathbb U-\mathbb S)\times\mathbb R\times\mathbb X$.

One way a singular Rosen metric arises in practice is illustrated by the following example.  Suppose $n=1$, and the spacetime is $\mathbb M_\beta = \mathbb R^3$ equipped with a Brinkmann metric
$$G_\beta = G_\beta(1) = 2\,du\,dv + x^2\,du^2 - dx^2.$$
Define a new coordinate system, as follows.  Let $\mathbb M_\rho=\mathbb R^3$, equipped with the singular Rosen metric
$$G_\rho = G_\rho(\sin u) = 2\,du\,dv - \sin^2\!u\,dx^2.$$
The new coordinates are defined by means of the smooth function $\phi:\mathbb M_\rho \to \mathbb M_\beta$, where
$$\phi(u,v,x) = (u, v + 2^{-1}x^2\cos u\sin u, x\sin u).$$
The set of singular points of the metric $G_\rho$ is the integer multiples of $\pi$: $\mathbb S=\mathbb Z\pi$.  Also, the function $\phi$ satisfies the following properties
\begin{itemize}
\item $\phi^*G_\beta = G_\rho$;
\item $\phi$ commutes with the standard dilation: $\phi\circ\mathcal D_t = \mathcal D_t\circ\phi$; and
\item $\phi$ is a one-to-one isometric diffeomorphism of $\mathbb M_\rho - \mathbb S$ onto its image, which is a dense open set in $\mathbb M_\beta$.
\end{itemize}

\begin{definition}\label{RosenUniverseDefinition}
  A {\em Rosen universe} is a Penrose limit $(M,G,\mathcal D)$, together with a singular Rosen plane wave $(\mathbb M_\rho,G_\rho,\mathcal D_\rho)$, and a smooth function $\phi:\mathbb M_\rho \to M$, such that:
  \begin{itemize}
  \item $\phi$ maps $\mathbb M_\rho-\mathbb S$ isometrically, diffeomorphically, and one-to-one onto a dense open subset of $M$; and
  \item $\phi$ is compatible with the dilation, i.e., $\phi\circ \mathcal D_\rho = \mathcal D\circ\phi$.
  \end{itemize}
\end{definition}

Thus, a little less formally, a Rosen universe is a manifold $\mathbb{M} = \mathbb{U} \times \mathbb{R} \times \mathbb{X}$, as above, equipped with a symmetric tensor field $G_\rho(h) = 2\,du\,dv - dx^Th(u) dx$, where $h$ is  smooth on $\mathbb{U}$, such that the endomorphism $h(u)$ is positive definite, for each $u \in \mathbb{U}$, except when  $u \in \mathbb{S}$, a discrete subset of $ \mathbb{U}$ and such that at each point of $\mathbb{S}$, the metric has a {\em removable singularity}.  A useful characterization of the removability is Theorem \ref{RosenUniverse}.

The first basic result is the theorem, whose proof we give below, independent of the work of Alekseevsky.
\begin{theorem}\label{AlekseevskyTheoremNew}
  Every plane wave is a Rosen universe.
\end{theorem}

The second basic result is the theorem characterizing the structure of the removable singularities of a Rosen universe.
\begin{theorem}\label{RosenUniverse}
  A metric $2\,du\,dv - dx^T h(u) dx$, defined on $\mathbb{M} = \mathbb{U} \times \mathbb{R} \times \mathbb{X}$, as above, with $h(u)$ smooth and symmetric and with $h(u)$ positive definite except at a point $u = u_0 \in \mathbb{U}$, has a removable singularity at $u_0$ if and only if there exists a (necessarily unique) symmetric endomorphism $K(u)$ of $\mathbb{X}$, smooth on $\mathbb{U}$, such that $K(u)h(u) = (u - u_0)^2 I$ and $K'(u_0) = 0$.
\end{theorem}
Here a prime denotes the derivative with respect to $u$ and $I$ denotes the identity endomorphism of $\mathbb{X}$. 

\subsection{Relations between the various forms of plane waves}\label{AlekseevskySection} 
Recall the family $\mathcal{A}$ metrics $G_\alpha(h)$, defined on $\mathbb{M} = \mathbb{U} \times \mathbb{R} \times \mathbb{X}$, as discussed above:
\[ G_\alpha(h) = \alpha \hspace{2pt} du - dx^Th(u) dx, \]
\[ \alpha = 2\,dv - b(u)v\,du + x^Tp(u)x\,du - x^Tq(u)dx.\]
Here $b$ is real valued, $h$, $p$ and $q$ take values in the endomorphisms of $\mathbb{X}$, with $h$ symmetric and positive definite. Also $b, p, q$ and $h$ are smooth on $\mathbb{U}$.

We note some transformations mapping the ensemble $\mathcal{A}$ to itself. In the following a prime is used to denote the derivative with respect to $u$.  
\begin{lemma}
  The collection $\mathcal{A}$ is invariant under $u$-dependent conformal transformations: if $(\mathbb{M}, G) \in \mathcal{A}$, then $(\mathbb{M}, e^{f(u)}G)\in \mathcal{A}$, where $f: \mathbb{U} \rightarrow \mathbb{R}$ is smooth. 
\end{lemma}
\begin{proof}
We put $V = e^{f(u)} v$ and $H(u) = e^{f(u)} h(u)$.  Then, using a prime to denote the derivative with respect to $u$,  we have the required expression:
\begin{align*}
  e^{f(u)} G_{\alpha}(h) &= \beta \,du - dx^TH(u) dx, \\
  \beta &=  e^{f(u)} \alpha \\
        &= 2\,dV - (f'(u) + b(u))V\,du  + e^{f(u)}x^Tp(u)x\,du  - e^{f(u)}x^Tq(u)dx, \\
        &=  2\,dV - B(u) V\,du  + x^TP(u)x\,du  - x^TQ(u) dx,
\end{align*}
  \[ e^{f(u)} G_{\alpha}(h)  = G_{\beta}(H), \quad B = b + f', \quad  P= e^{f} p, \quad  Q = e^{f}q.\]
\end{proof}

\begin{lemma}
  The collection $\mathcal{A}$ is invariant under linear transformations on the vector-valued variable $x \in \mathbb{X}$, $x \rightarrow X = L^{-1}(u)(x)$, where $L$ is smooth on $\mathbb{U}$ and takes values in the automorphisms of $\mathbb{X}$, so $\det(L)$ is non-zero everywhere on $\mathbb{U}$.
\end{lemma}

\begin{proof}
Put $x = L(X)$,  $L' = SL$, where $S = L' L^{-1}$ is smooth, so $dx = S(x) du + L dX$ and put $H = L^Th L$.  Then we have:
\begin{align*}
  G_\alpha(h) &= \alpha du  - (dX^TL^T + x^T S^Tdu)h(L\, dX + Sx\, du), \\
  G &= \beta\,du - dX^T H\, dX,\\
  \beta  &= 2\,dv - bv\,du  + x^T (p - qS - S^T hS) x\, du - x^T (q  + 2S^Th )L\, dX,
  \end{align*}
  \[  G_\alpha(h) = G_\beta (H), \quad  \beta = 2\,dv - bv\,du + X^TP X\, du - X^T Q\, dX, \]
  \[ P = L^T(p - qS - S^T hS)L, \quad Q = L^T(q  + 2S^Th)L .\]
\end{proof}
\begin{lemma}
  The collection $\mathcal{A}$ is invariant under smooth diffeomorphisms of the $u$-variable.
\end{lemma}

\begin{proof}
We put $U = U(u), \quad  v =  U'(u)V, \quad x = x$, where $U: \mathbb{U} \rightarrow \mathbb{\tilde{U}}$ is a smooth bijection.  We use a dot to denote differentiation with respect to the variable $U$, so, in particular,  we have $du = \dot{u}dU$ and $d(\dot{u}) = \ddot{u}dU$.   Note that for a diffeomorphism the smooth function $\dot{u} = (U')^{-1}$ is nowhere vanishing.   Also $\ddot{u} = - (U')^{-3}U''$.
Then we get:
\begin{align*}
  G &= \beta\, dU -  dx^T H(U)dx, \\
  \beta &=  \dot{u}\alpha  \\
    &= 2\,dV - B(U)V\, dU + x^TP(U)x\, dU - x^TQ(U)\,dx
\end{align*}
\[ B(U)= \dot{u} b(u) + \ddot{u}(\dot{u})^{-1}, \quad P(U) = (\dot{u})^2p(u), \quad Q(U) = \dot{u}\omega(u),  \quad  H(U) = h(u).\]
We can refine the collection $\mathcal{A}$ by writing $\mathcal{A}_{\mathbb{U}}$, for all Alekseevsky metrics on $\mathbb{U} \times \mathbb{R} \times \mathbb{X}$, with $\mathbb{U}$ a connected open subset of the reals.  Then Lemmas $4$ and $5$ map  $\mathcal{A}_{\mathbb{U}}$ to itself, whereas Lemma $6$ maps $\mathcal{A}_{\mathbb{U}}$ to $\mathcal{A}_{\tilde{\mathbb{U}}}$, where $u(\mathbb{\tilde{U}}) = \mathbb{U}$.
\end{proof}

\begin{theorem}\label{PlaneWavesSame}
Each plane wave spacetime of the collection $\mathcal{A}$ is isomorphic to a metric of each of the types $\mathcal{A}_\alpha$ and $\mathcal{A}_\beta$ and,   at least locally in the $u$-variable,  is isometric to a metric of the type $\mathcal{A}_\rho$.    The same is true if we demand also that the central null geodesic be preserved, i.e. these isometries are also isomorphisms of the corresponding Penrose limits.
\end{theorem}
We prove a series of Lemmas which will give us the first part of the theorem.   The transformations used in the Lemmas explicitly keep the central null geodesic invariant, giving the second part.

We say that two types coincide, if given a spacetime of one type, there is a spacetime of the other type, such that there is an isometry mapping one to the other, preserving the standard dilation groups.
\begin{lemma}
 Types $\mathcal{A}$ and $\mathcal{A}_\alpha$  coincide.
\end{lemma}
\begin{proof}
Let $G$ be an element of $\mathcal{A}$ written in the standard form, as  given above:
\begin{align*}
  G &= \alpha\, du - dx^T h(u)dx, \\
  \alpha &= 2\,dv - b(u)v\,du  + x^T p(u)x\, du  - x^Tq(u)\,dx.
\end{align*}
First we smoothly factorize  the positive definite automorphism $h(u)$:
\[ L^T(u)h(u)L(u)  = I.\]
Here  $L$ is smooth on $\mathbb{U}$ and takes values in the automorphisms of $\mathbb{X}$.  It is standard that such a smooth factorization exists.
Then we put: $x = L(u)(X)$.  Using Lemmas $5$ and $2$, we obtain a metric of type $\mathcal{A}_\alpha$, the point being that in the language of Lemma 2, the endomorphism $H$ is just the identity and that $H$ is preserved, when employing Lemma 2. Conversely any $\mathcal{A}_\alpha$ metric is trivially of type $\mathcal{A}$, so we are done.
\end{proof}
\begin{lemma}
 Types $\mathcal{A}_\alpha$ and $\mathcal{A}_\beta$  coincide.
\end{lemma}
\begin{proof}Begin with  a general $\mathcal{A}_\alpha$  metric:
\[  G = \alpha\, du  - dx^Tdx, \quad \alpha = 2\,dv - b(u)v\,du - x^Tq(u)dx + x^Tp(u)x\, du.\]
First we note first  that the terms $du(2dv - b(u)\,du)$ in $G$ can be reduced to just $2\,du\,dv$ by the replacements:
\[ v = e^{f(u)} V,  \quad   u = u(U), \quad  f'(u) = 2^{-1} b(u),   \quad dU = e^{f(u)}du. \]
Then the metric becomes:
\[ G = \beta\, dU - dx^Tdx,  \quad  \beta = 2\,dV - x^TQ(U)\,dx + x^TP(U)x\,dU, \]
\[  P(U) = p(u) (U')^{-2}, \quad Q(U) = q(u) (U')^{-1}.\]
In particular, $P$ is symmetric, because $p$ is symmetric, whereas $q$ is skew, because $q$ is skew.
Dropping capital letters, we have shown that,  without lose of generality,  we may write any  $\mathcal{A}_\alpha$ metric in the form:
\[  G_\alpha(h) = \alpha\, du  - dx^Tdx, \quad\alpha = 2\,dv - x^Tq(u)dx + x^Tp(u)x\, du.\]
Here $p$ is symmetric, whereas $q$ is skew.
Next we transform the variable $x$, writing $x = L(u)(X)$, with $L(u)$ an orthogonal transformation, so $L^T(u)L(u) = I$, for each $u \in \mathbb{U}$ and $L$ is smooth on $\mathbb{U}$.   Also we put $L' = SL$, where $S$ is skew.  Using Lemma 2, since $h = H = I$, the transformed $q$ is:
\[  Q = L^T(q  - 2Sh )L.\]
First we want to choose $S$, such that $Q$ is symmetric.   So we need to solve for a skew endomorphism $S(u)$ the equation:
\[ S(u) h(u) + h(u)S(u) =  q(u).\]
After conjugating with a suitable smooth orthogonal transformation, we may assume, locally at least in $u$,  that $h$ is a diagonal matrix with positive real-valued entries $h_{ii}(u) = \lambda_i(u)$.  Then the equation $Sh + hS = q$ written out in terms of components is:
\[ (\lambda_i + \lambda_j)S_{ij} = q_{ij}.\]
Since $q_{ij}$ is skew, and since $\lambda_i + \lambda_j > 0$, for any $i$ and $j$, there is a unique skew local smooth solution for $S(u)$.  By uniqueness, the local solutions patch to give a unique global skew solution,  $S(u)$.  Then after integrating to get the orthogonal transformation $L(u)$, we end with the metric of the following form:
\[ G = \beta du - dX^T dX, \]
\[ \beta = 2\,dv - X^TQ(u)dX + X^TP(u)X du.\]
Here $Q(u)$ is symmetric and smooth and, without loss of generality $P$ is also symmetric.  Finally we put $v = V + 4^{-1}X^T Q(u)X$, giving:
\[ \beta = 2\, dV + X^T \tilde{P}(u)X du, \quad \tilde{P} = P + 4^{-1}Q'.\]
Since $\tilde{P}$ is symmetric, we obtain a standard Brinkmann metric:
\[ G = 2 \,du\,dv + X^T \tilde{P}(u)X \,du^2 - dX^T dX.\] 
So we are done, the converse being trivial,  as before,  and the Lemma is proved.\end{proof}
Note that the Brinkmann metric is generally on a manifold $\tilde{\mathbb{U}} \times \mathbb{R} \times \mathbb{X}$, rather than the original $\mathbb{U} \times \mathbb{R} \times \mathbb{X}$, since, during the proof of Lemma 8, the original $u$ coordinate was replaced by an appropriate function of $u$.

We complete the proof of the Theorem with a final Lemma:
\begin{lemma}\label{EveryBrinkmannRosen}
  Each metric of type $\mathcal{A}_\rho$ is a metric of type $\mathcal{A}_\beta$.  Conversely each metric of type $\mathcal{A}_\beta$,  at least locally in the $u$-variable, is a metric of type $\mathcal{A}_\rho$.
\end{lemma}
\begin{proof} Each spacetime of type $\mathcal{A}_\rho$ is trivially a metric of type $\mathcal{A}$ and each  such is a  metric of type $\mathcal{A}_\beta$, by our earlier Lemmas, so the first statement of the Lemma follows immediately.  The proof of the second statement is more involved.  

So let a Brinkmann metric be given:
\[ G  =  2\,du\,dv + x^Tp(u)x \,du^2 - dx^T dx. \]
Here the smooth endomorphism $p$ of $\mathbb{X}$ is symmetric. We again put $x = L(u)(X)$, where $L$ is smooth in $u$ and invertible.  Also put $L'(u) = S(u)L(u)$, with the endomorphism $S(u)$ of $\mathbb{X}$ smooth on $\mathbb{U}$ and symmetric.  Also put $h(u) = L^T(u)L(u)$, so $h(u)$ is smooth symmetric and positive definite.  Then the metric becomes:
\begin{align*} 
 G  &=  2\,du\,dv + x^T p(u)x\,du^2 - dx^T dx \\
    &= 2\, du\, dv + x^T\left(p - S^2\right)x\,du^2  - 2x^TSL\,dX\,du -  dX^Th\,dX.
\end{align*}
Now put $v = V + x^TS x$.  Then we get:
\[  G = 2\,du\,dV + x^T(S'  + S^2 + p)x\,du^2 -  dX^Th\,dX\]
Finally we need $S$ to obey the Sachs equation:
\[ 0 =   S' + S^2 + p.\]
Then we get precisely the required form of a Rosen metric:
\[ G = 2\, du\,dV -  dX^Th\,dX.\]
So, given $p$,  we need to solve the pair of equations, for smooth endomorphisms $S$ and $L$, with $S$ symmetric and $L$ invertible:
\begin{align*}
  0 &=   S' + S^2 + p\\
  L' &= SL.
\end{align*}
Differentiating the second equation we get the linear second-order homogeneous differential equation:
\[ L'' + pL = 0.\]
To solve, we first solve the equation $L'' + pL = 0$, with initial conditions at some basepoint  $u = u_0$, such that $L(u_0)$ is invertible and $L'(u_0) = S_0L(u_0)$ with $S_0$ symmetric. For example, it suffices to take $L(u_0) = I$,  the identity automorphism of $\mathbb{X}$ and  $L'(u_0) = 0$, with $S_0 = 0$.   Given such initial conditions, it is standard that there is then a smooth solution for $L(u)$, defined for all $u \in \mathbb{U}$,  which obeys the desired initial conditions.

There is then a discrete collection of points, $\mathbb{S} \subset \mathbb{U}$, empty, finite, or countably infinite, which we call the degeneracy set of  the solution $L(u)$,  such that $\det(L(u)) \ne 0$, for all $u$ in the  open set $ \mathbb{U}- \mathbb{S}$, whereas $L(s)$ is not invertible at any $s \in \mathbb{S}$.    Finally the endomorphism $S$ is defined on the set $\mathbb{U} - \mathbb{S}$  by the formula $S(u) = L'(u)L^{-1}(u)$, for any $u \in \mathbb{U} - \mathbb{S}$.    Then $S$ is symmetric and all the required equations are satisfied on the set $\mathbb{U} - \mathbb{S}$.  
\end{proof}
The last part of the proof of Theorem 3 also gives:
\begin{theorem} Each plane wave gives rise naturally to a Rosen universe.
\end{theorem}
\begin{proof}  Theorem 3 shows that any plane wave may be put in Brinkmann form.    Then the last part of Theorem 3 constructs a Rosen universe with metric $2\, du\,dV -  dX^TL^T(u)L(u)dX$, with singularities at the isolated points where the endomorphism  $L(u)$ has vanishing determinant, isometric to the Brinkmann plane wave away from the singularities.   Each such singularity is removable, because if $\det(L(u_1)) = 0$, for some $u_1 \in \mathbb{U}$, we can simply redo the transition from the given Brinkmann  to Rosen, but now taking initial conditions at $u = u_1$, such that $L(u_1)$ is the identity. The resulting Rosen metric is guaranteed to be non-singular in an open sub-interval of $\mathbb{U}$ containing the point $u_1$.
\end{proof}

\section{Jacobi and Sachs equations}
In this section, we describe the Jacobi and Sachs equations of a spacetime, and derive their basic relations and properties. This system of equations governs the coordinate change from Brinkmann to Rosen, as we shall discuss in \S\ref{AlekseevskySection}. The section concludes with the cases of Brinkmann and Rosen plane waves, and proves a result (equivalent to one) of Alekseevsky \cite{Alekseevsky} characterizing plane waves that we rely on extensively in the sequel.

Let $(M,g)$ be a smooth spacetime of dimension $n+2$.  By $T^*M'$ denote the cotangent bundle of $M$ with the zero section deleted.  This carries a canonical one-form, which we denote by $\theta$, defined in local coordinates by $\theta=p_idx^i$, where $p_i$ are the fibre coordinates dual to $\partial/\partial x^i$.  Thus we have $\theta_\alpha(v) = \alpha(\pi_*v)$ where $\alpha$ is a point of $T^*M'$, $v$ is a tangent vector to $T^*M'$ at $\alpha$, and $\pi:T^*M'\to M$ is the bundle projection.  The two-form $\omega=d\theta$ is a symplectic form on $T^*M'$.  It induces a Poisson structure: for smooth complex-valued functions $f,g$ on $T^*M'$,
$$\{f,g\} + \omega^{-1}(df,dg) = 0.$$
In terms of local coordinates, then
$$\{f,g\} = \sum_i \left(\frac{\partial f}{\partial x^i}\frac{\partial g}{\partial p_i} - \frac{\partial f}{\partial p_i}\frac{\partial g}{\partial x^i}\right).$$
Since $\{f,g\}$ is linear with respect to constants, and satisfies the Leibniz rule in each argument, it follows that the operator $\{f,-\}:C^\infty(T^*M')\to C^\infty(T^*M')$ is induced by a vector field, which we call the {\em symplectic gradient}, and write
$$\nabla^\omega f := \{f,-\} = \frac{\partial f}{\partial x^i}\frac{\partial}{\partial p_i} - \frac{\partial f}{\partial p_i}\frac{\partial}{\partial x^i}.$$

Let $E$ be the vector field generating the scaling group $\epsilon_t:p\mapsto e^tp$ on sections $p$ of $T^*M$.  In local coordinates, $E=p_i\partial/\partial  p_i$.  Then $\mathscr L_E\theta=\theta$.

Let $H:T^*M'\to\mathbb R$ be the geodesic Hamiltonian, so $H$ is homogeneous of degree two and its restriction to each fibre of the cotangent bundle is a quadratic form of signature $(1,n+1)$, which varies smoothly from point to point.  In local coordinates, $2H=g^{ij}p_ip_j$.  The null cone bundle $\mathcal N\subset T^*M'$ is the zero set of $H$, a submanifold of dimension $2n+3$.  Since $H$ is homogeneous of degree two, $E(H)=2H$.  Therefore $E$ is tangent to $\mathcal N$.  Also, because $\mathcal N$ is an energy surface of $H$, the symplectic gradient $\nabla^\omega H$ is tangent to $\mathcal N$.  Moreover, we have $[E,\nabla^\omega H]=\nabla^\omega H$, so the Lie algebra generated by $E,\nabla^\omega H$ is the non-abelian two-dimensional nilpotent Lie algebra.

The integral curves of $\nabla^\omega H$ are the geodesics.  The integral curves tangent to $\mathcal N$ are the null geodesics.  The ``space of null geodesics'' is the set of integral curves $\nabla^\omega H$ in $\mathcal N$.  (We shall assume that this flow acts properly so that the space of null geodesics is a smooth manifold.  This is always true locally on $M$.)  The space of {\em unparameterized} null geodesics is the space of integral manifolds of the distribution $\mathfrak A$ spanned by $E$ and $\nabla^\omega H$.  Let $\mathcal{GN}$ be the space of null geodesics, and $\mathcal{PN}$ the space of unparameterized null geodesics.  So $\mathcal{GN}$ is the quotient of $\mathcal N$ by $\nabla^\omega H$, a smooth manifold of dimension $2n+2$, and $\mathcal{PN}$ is the quotient of $\mathcal N$ by $\mathfrak A$, a smooth manifold of dimension $2n+1$.  Then $\mathcal{GN}$ is a $\mathbb R_+$ bundle over $\mathcal{PN}$, whose fibres are the orbits of $E$.

The symplectic form $\omega$ descends to a symplectic form on $\mathcal{GN}$.  The canonical one-form $\theta$ descends to a contact structure $\mathbb J=\ker\theta$ on $\mathcal{PN}$; so $\mathbb J\subset T(\mathcal{PN})$ is a codimension one subbundle (everywhere non-integrable as a distribution of vector fields).  The form $\omega$ is naturally the symplectization of the contact structure $\mathbb J$.

\subsubsection*{Summary}
Our basic objects are:
\begin{itemize}
\item $(M,g)$ is a spacetime of dimension $n+2$, where $n\ge 1$;
\item $\mathcal N$ is the null cone bundle of $M$, which is a quadric in the cotangent bundle, of dimension $2n+3$;
\item $\mathcal{PN}$ is the space of (unparameterized) null geodesics, of dimension $2n+1$;
\item the vector field $E$ is the homogeneity operator, and the symplectic gradient $\nabla^\omega H$ is the Hamiltonian vector field;
\item together, $E,\nabla^\omega H$ span a two-dimensional Lie algebra $\mathfrak A$, such that $\mathcal{PN}$ is the quotient of $\mathcal N$ by $\mathfrak A$;
\item finally, $\mathbb J$ is a (contact) distribution of $2n$ planes on $\mathcal{PN}$, and $\mathbb J$ carries a symplectic form $\omega$ up to scale.
\end{itemize}

\subsection{Jacobi fields}

A Jacobi field is a tangent vector field to $\mathcal{PN}$.  That is, a vector field $X$ on $\mathcal N$, modulo $\mathfrak A$, that normalizes the distribution $\mathfrak A$.  We shall normally be interested in tangent vectors to $\mathcal{PN}$, that is, Jacobi fields evaluated at a point $\gamma\in\mathcal{PN}$, or equivalently vector fields along a null geodesic $\gamma$ (regarded as an integral surface in $\mathcal N$) modulo $\mathfrak A$ that normalize the distribution $\mathfrak A$ along $\gamma$.

We now relate these notions to the usual constructions in differential geometry.  We introduce abstract indices so that $\theta=p_a\theta^a$, $\omega=D\theta=Dp_a\wedge\theta^a$ where $Dp_a$ is the Levi-Civita connection.  Then $\nabla^\omega H$ is $g^{ab}p_a\nabla_b$.  A Jacobi field $J$ as we have defined it is a vector field on $\mathcal N$.  Therefore it can be written (in indices) as
\begin{equation}\label{JacobiAbstractIndex}
  J = \alpha^a\nabla_a + \beta_a\partial^a
\end{equation}
where we put $\partial^a=\partial/\partial p_a$ the generator of the translations up the fibre, and $\nabla_a$ is the horizontal distribution of the Levi-Civita connection.  Together these operators obey the structure
$$[\nabla_a,\nabla_b]={R_{abc}}^dp_d\partial^c,\quad [\nabla_a,\partial^b]=0,\quad [\partial^a,\partial^b]=0$$
The variables $p_a$ satisfy $[p_a,p_b]=0$, and
$$[\partial^a,p_b]=\delta^a_b, \quad [\nabla_a,p_b]=0.$$
For $J$ to be a Jacobi field, we need $J$ to commute with $E$ and $[p^b\nabla_b,J]$ to vanish modulo $E=p_a\partial^a$ and $\nabla^\omega H = p^a\nabla_a$.  We can in fact make this commutator equal to zero:
\begin{align*}
  [p^a\nabla_a,J] &= [p^a\nabla_a,\alpha^b\nabla_b + \beta_b\partial^b] \\
                  &= (p^a\nabla_a\alpha^b)\nabla_b + p^a\alpha^b[\nabla_a,\nabla_b] + (p^a\nabla_a\beta_b)\partial^b - g^{ab}\beta_a\nabla_b\\
                  &= (\dot\alpha^b - \beta^b)\nabla_b + (p^a\alpha^b{R_{abc}}^dp_d + \dot\beta_c)\partial^c = 0.
\end{align*}
(The operator $p^a\nabla_a$ is differentiation along the Hamiltonian flow of $H$, so we denote it with a dot.)
We have proved:
\begin{lemma}
  Let $J = J_H + J_V$ be the decomposition of a Jacobi field along $\gamma$ into horizontal and vertical components using the Levi-Civita connection.  Then $\dot J_H=J_V$ and $\dot J_V = PJ_H$ where $PX = R(\dot\gamma, X)\dot\gamma$ is the tidal curvature operator.  In particular, we have the {\em Jacobi equation}:
  $$\ddot J_H = PJ_H.$$
\end{lemma}
Note that the equation $\ddot J_H=PJ_H$ is compatible with the equivalence relation of $J\mod\mathfrak A$.

\subsection{Abreastness}

A Jacobi field is called {\em abreast} if it is in the contact distribution $\mathbb J$.  Let $\gamma\in\mathcal{PN}$ be a null geodesic, and $\mathbb J_\gamma$ the space of abreast Jacobi fields along $\gamma$.  Then $d\theta$ defines a symplectic structure up to scale on $\mathbb J_\gamma$.  Indeed, the transition between two sections of the bundle $\mathcal{GN}\to\mathcal{PN}$ is of the form $d(\lambda\theta) = \lambda d\theta + d\lambda\wedge\theta$, but the second term is zero on $\mathbb J_\gamma$.

The construction of $\mathcal{PN}$ is conformally invariant.  If $H$ is replaced by $\lambda H$ for $\lambda>0$, then $\nabla^\omega H$ becomes $\lambda\nabla^\omega H + H\nabla^\omega\lambda$, but the latter term vanishes on $\mathcal N$ because $H=0$ there.  Similarly, the definition of Jacobi fields and their abreastness is conformally invariant.  We distill these observations as follows:
\begin{lemma}\label{ConformalInvariance}
  The space $\mathbb J_\gamma$ of abreast Jacobi fields is conformally invariant.  Moreover, the symplectic form on $\mathbb J_\gamma$ defined by
  $$\omega_g(X,Y) = g(X,\dot Y) - g(\dot X,Y) $$
  is conformally invariant (up to an overall constant depending on the affine parameterization of $\gamma$).
\end{lemma}
\begin{proof}
The first part follows from our general observations, but we can also give a direct proof.  Let $\hat g=e^{2\varphi}g$,
$$\hat\nabla_XY = \nabla_XY + d\varphi(X)Y + d\varphi(Y)X - g(X,Y)\nabla\varphi.$$
Also, with
$$R(X,Y,Z,W) = g(R(X,Y)Z,W),$$
$$T(X,Y) = \nabla^2\varphi(X,Y) - X(\varphi)Y(\varphi) + \frac12 |d\varphi|^2g(X,Y) $$
and
$$(g\owedge T)(X,Y,Z,W) = g(X,Z)T(Y,W) + g(Y,W)T(X,Z) - g(X,W)T(Y,Z)-g(Y,Z)T(X,W). $$
$$\hat R = e^{2\varphi}\left( R + g\owedge T\right),$$
Recall that the tidal curvature is\footnote{These conventions are selected so that the tidal curvature has positive trace on a geodesic where the energy density is positive.



}
$$P(X,Z) = R(\dot\gamma,X,\dot\gamma,Z).$$
Suppose that $g(X,\dot\gamma)=g(Y,\dot\gamma)=0$.  Then
\begin{align*}
  (g\owedge T)(X,\dot\gamma,Z,\dot\gamma)
  &= g(X,Z)T(\dot\gamma,\dot\gamma)\\
  &=g(X,Z)\left(\ddot\varphi - \dot\varphi^2\right).
\end{align*}
Suppose $J$ is an abreast Jacobi field, so $g(J,\dot\gamma)=0$, and 
$$\ddot J = (g^{-1}P)J$$
modulo multiples of $\dot\gamma$.

Let $u$ be an affine parameter along $\gamma$, so that the dot is differentiation with respect to $u$.  Define an affine parameter for $\gamma$ in the rescaled metric $\hat g$ has $dU=e^{2\varphi}du$.  Let a prime denote the covariant differentiation $\hat D/dU$.  So we have $\gamma'=e^{-2\varphi}\dot\gamma$, $\gamma'(e^{-2\varphi})=-2\dot\varphi\,e^{-2\varphi}$.

Defining $\widehat P(X,Z) = \widehat g(\widehat R(\gamma',X)\gamma',Z)$, we have
$$\widehat P(X,Z) = e^{-2\varphi}\left(P(X,Z) + g(X,Z)(\ddot\varphi -\dot\varphi^2)\right).$$
Thus we have the operator:
$$\hat g^{-1}\widehat P = e^{-4\varphi}(g^{-1}P + (\ddot\varphi - \dot\varphi^2)I) $$
We must show that $J'' = (\hat g^{-1}\hat P)J$ modulo $\gamma'$.  We have
$$J' = e^{-2\varphi}(\dot J + \dot\varphi J) $$
\begin{align*}
  J'' &=  e^{-4\varphi}( -2\dot\varphi(\dot J + \dot\varphi J) + \ddot J + \ddot\varphi J + 2\dot\varphi\dot J + \dot\varphi^2 J)\\
      &= e^{-4\varphi}(\ddot J + (\ddot\varphi - \dot\varphi^2)J)\\
      &= e^{-4\varphi}((g^{-1}P)J + (\ddot\varphi - \dot\varphi^2)J)\\
      &= (\hat g^{-1}\widehat P)J.
\end{align*}

We can now prove the second part, using the same formalism as the first part:
$$\widehat\omega(X,Y) = e^{2\varphi}[g(X,Y') - g(X',Y)] = g(X,\dot Y + \dot\varphi Y) - g(\dot X + \dot\varphi X) = \omega(X,Y).$$

\end{proof}

\subsection{The shear tensor}
Let $(\mathbb J,\omega)$ be the space of abreast Jacobi fields, a vector space of dimension $2n$ together with a symplectic form $\omega$ (defined up to an overall constant).  A linear map $L:\mathbb R^n\to\mathbb J$ is called a {\em Lagrangian matrix} if $L$ has rank $n$ and $\omega(L\otimes L)=0$.  Let $\mathcal L(\mathbb J)$ be the space of Lagrangian maps to $\mathbb J$.  The {\em Lagrangian Grassmannian} $\mathcal{LG}(\mathbb J) = \mathcal{L}(\mathbb J)/\op{GL}(n,\mathbb R)$ is the quotient of the Lagrangian maps under the right action of $\op{GL}(n)$.  We shall discuss the problem of how to construct the points of the Lagrangian Grassmannian, starting from the Jacobi equation.

Fix a null geodesic $\gamma$ in $M$, with affine parameter $u$ and differentiation denoted by a dot.  The Jacobi equation for a section of $\gamma^{-1}TM$ is
$$\ddot J(u) = P(u)J(u) $$
where $P(u) = R(\dot\gamma(u),-)\dot\gamma(u)$ is the tidal curvature operator (a linear endomorphism of $T_{\gamma(u)}M$ that is symmetric with respect to $g$).  Suppose $L$ is a Lagrangian matrix, so the columns of $L$ are $n$ Jacobi fields on $M$ that are mutually isotropic with respect to $\omega$, and such that $L\oplus\dot L$ has rank $n$.  Then we have, with obvious notation,
$$\ddot L(u) = P(u)L(u),$$
and, denoting the columns of $L$ by $L_i$,
\begin{equation}\label{LIsotropicExplicit}
  \omega(L_i,L_j) = g(L_i,\dot L_j) - g(L_j,\dot L_i) = 0.
\end{equation}
\begin{lemma}
  In a neighborhood of any point where $L(u)$ is an invertible linear map, the {\em shear operator}
  $$S = \dot L L^{-1} $$
  satisfies the {\em Sachs equation}
  \begin{equation}\label{SachsEquation}
    \dot S + S^2 - P = 0.
  \end{equation}
  The shear operator is symmetric with respect to $g$ (i.e., $g(S\otimes I)=g(I\otimes S)$).  Conversely, given a smooth symmetric solution to \eqref{SachsEquation} in any interval around $u=u_0$, there exists a Lagrangian matrix such that $\dot L = SL$.
\end{lemma}
\begin{proof}
  Taking a derivative of $S=\dot LL^{-1}$, we obtain at once
  $$\dot S = \ddot L L^{-1} - \dot L L^{-1} \dot L L^{-1} = P - S^2,$$
  as required.  The symmetry of $S$ follows from the relation of isotropy, for $x,y\in\mathbb R^n$:
  $$g(\dot Lx,Ly) = g(Lx,\dot Ly).$$
  With $X=Lx, Y=Ly$, this implies that $S$ is a symmetric endomorphism.
  
  For the converse, the equation $\dot L=SL$ admits a local solution with $L$ invertible (unique modulo $GL(n)$).  Differentiating this equation gives
  $$\ddot L = \dot SL + S\dot L = (-S^2+P)L + S^2L = PL.$$
  We now have a second order differential equation, whose solutions are uniquely determined from initial conditions $L(u_0),\dot L(u_0)$.  Moreover, $L$ is isotropic because near $u_0$ we have \eqref{LIsotropicExplicit}:
  $$\omega(L_i,L_j) = g(L_i,\dot L_j) - g(L_j,\dot L_i) = g(L_i,SL_j) - g(L_j,SL_i) = 0$$
  by the symmetry of $S$.  Since this vanishes locally, it must also vanish globally because of uniqueness of solutions to second order equations.
\end{proof}

  Let $\widehat g = e^{2\varphi}g$ be a conformally rescaled metric, and $L$ a Lagrangian matrix.  Then, in the notation of the proof of Lemma \ref{ConformalInvariance}, we have:
  $$\widehat S = L'L^{-1} = e^{-2\varphi}(\dot LL^{-1} + \dot\varphi).$$
  In particular:
  \begin{lemma}
    The trace-free parts of the $(2,0)$-tensors $gS$ and $\widehat g\widehat S$ are identical, and there is a unique conformal factor $e^{2\varphi}$ up to a multiplicative constant such that $\widehat S$ is trace-free.
  \end{lemma}

  \subsection{Jacobi fields and Sachs equations on plane waves}
  We develop the ideas in this section in the examples of the Brinkmann and Rosen metrics.

  For a plane wave in the Brinkmann form $G_\beta(p)$ in \eqref{Brinkmann1}, it can be readily checked that the tidal curvature is $-p(u)$, and therefore the equation of Jacobi transport of a vector $X:u\to\mathbb X$ along the geodesic $v=x=0$ is
$$\ddot X + p(u)X = 0.$$
Now $X$ is an ordinary vector in the $n$-dimensional Euclidean space $\mathbb X$, $p(u)$ is a symmetric matrix, and the dot is (as usual) differentiation with respect to $u$.  The space $\mathbb J$ of Jacobi fields $X$ is a real vector space of dimension $2n$, being determined by a pair of initial conditions $X_0,\dot X_0$ in $\mathbb X$.  The symplectic form for a pair of Jacobi fields $X,Y\in\mathbb J$ is just
$$J(X,Y) = X^T\dot Y - Y^T\dot X.$$
Note that $\frac{d}{du}\omega(X,Y)=0$, because $X,Y$ satisfy the same second order equation, and $p$ is symmetric.  The Sachs equation for a smooth function $S:\mathbb U\to\mathbb X\odot\mathbb X$ is just
$$ \dot S + S^2 + p = 0.$$
This form of the Sachs equation shall be used in the construction of associated Rosen metrics in the next section.

For a Penrose limit, we can consider Jacobi fields along the central null geodesic.  The following theorem gives two basic examples:
\begin{theorem}\label{JacobiPlaneWaveTheorem}\mbox{}
  \begin{enumerate}[(a)]
  \item For a Rosen plane wave \eqref{Rosen1}, a basis of the Jacobi fields is
    $$\partial_x, \quad x^T\partial_v + H\cdot\partial_x.$$
    (And $\partial_v$ is an additional non-abreast Jacobi field.) The tidal curvature tensor of the Rosen metric is
  $$P = dx^T\left(\tfrac12\ddot g - \tfrac14\dot g g^{-1}\dot g\right)dx.$$
  \item For the Brinkmann plane wave \eqref{Brinkmann1}, the Jacobi equation is
    $$\ddot X + p(u) X = 0,$$
    for the $\mathbb R^n$-valued function $X$; the Sachs equation is
    $$\dot S + S^2 + p(u) = 0 $$
    for the symmetric $n\times n$ matrix $S=S(u)$.  The tidal curvature is thus $P(u) = -dx^Tp(u)dx$.  Moreover, the hypersurfaces $u=$constant are affine spaces (and therefore totally geodesic).
  \end{enumerate}
  In both cases, the space $\mathbb J$ of (abreast) Jacobi fields is a $2n$-dimensional vector space, where $n=\dim\mathbb X$.
\end{theorem}
\begin{proof}
  The first statement in (a) follows immediately from the fact that these vector fields are Killing vectors of the metric, and therefore Jacobi fields.
To compute the curvature of a Rosen metric, the passage from Brinkmann to Rosen was mediated by the introduction of a Lagrangian matrix $L(u)$, such that $\ddot L +pL=0$ and $S=\dot L L^{-1}$ is symmetric, and satisfies the Sachs equation $\dot S + S^2 + p=0$.  The Rosen metric in the resulting Jacobi frame is then
$$2\,du\,dv - dx^Tg(u)dx,\qquad g(u) = L^TL.$$
  We have
  \begin{align*}
    \dot g &= 2L^TSL\\
    \ddot g &= 2L^T(\dot S + 2S^2)L = \tfrac12\dot gg^{-1}\dot g -2L^TpL.              
  \end{align*}
  Rearranging, and recalling $P=-dx^Tpdx$ now gives the theorem.

  For (b), we must show that the tidal curvature of the Brinkmann spacetime is $dx^TP(u)dx$ (the signs follow because the metric is negative definite on $\mathbb X$).  To do this, we have the connection
  \begin{align*}
    D\otimes(du) &= 0\\
    D\otimes(dv) &= -\tfrac12 x^T\dot p x\, du\otimes du - x^T p(dx\otimes du + du\otimes dx)\\
    D\otimes(dx) &= -px\,du\otimes du.
  \end{align*}
  So $\partial_v,\partial_x$ comprises a parallel frame on every $u=$constant slice.  The remaining assertion now follows from:
  $$D^2\otimes(\partial_x) = (du\wedge p\,dx)\otimes \partial_v.$$

\end{proof}

Note that the tidal curvature of the Rosen metric can be equivalently rewritten, with $H$ a primitive of $g^{-1}$, as follows:
$$-2P=dx\dot H^{-1}\left(\dddot H - \frac32\ddot H\dot H^{-1}\ddot H\right)\dot H^{-1}dx^T = dx(-\ddot g + \frac12 \dot gg^{-1}\dot g)dx^T.$$
Therefore we have:
\begin{corollary}
  The tidal curvature of a Rosen metric satisfies
  $$\mathcal S(H) = -2P$$
  where $\mathcal S(H)$ is the {\em Lagrangian Schwarzian} of $H$:
  $$\mathcal S(H) = dx^T\dot H^{-1}\left(\dddot H - \frac32\ddot H\dot H^{-1}\ddot H\right)\dot H^{-1}\,dx.$$
\end{corollary}
The full implications of this Corollary shall be deferred until a later article in this series, however the main point is that $L(u)$ (modulo $\op{GL}(n)$) is a positive curve in the Lagrangian Grassmannian, and the curvature serves as a natural model of the Lagrangian Schwarzian of the curve.

\section{Characterization of plane waves}\label{CharacterizationSection}
The following result is equivalent to the Theorem of Alekseevsky \cite{Alekseevsky}, discussed in \S\ref{AlekseevskySection}, but we include a rather different proof than his.  Because this result is so fundamental in our work, we thought it useful to provide an independent proof, even though it is an easy consequence of results already in the literature.  The proof we give also contains within it a very general construction of Penrose limits in the sense of \cite{penrose1976any}, although we shall not elaborate this construction because it is not used in the present paper.
\begin{theorem}\label{PlaneWaveCharacterization}
  Let $(M,g,\mathcal D)$ be a plane wave spacetime.  Then there exists (locally in $u$) a Rosen metric for $M$.
\end{theorem}

We shall prove this by essentially showing that a plane wave with a dilation is equal to its Penrose limit (a Rosen metric).  The proof uses ideas from metric (sub-Riemannian) geometry.  A general reference is \cite{Montgomery}, and for the specific details \cite{Bellaiche1996} and \cite{gromov1996carnot}.  

A dilation in a metric space $(\mathcal U,d)$ is defined as a one-parameter group $t\mapsto \delta_t$ of homeomorphisms, fixing a point, such that $d(\delta_tx,\delta_ty)=e^td(x,y)$ for all $x,y\in\mathcal U$.  A sub-Riemannian contact manifold is a connected  manifold $(\mathcal U,\theta)$, with a positive-definite metric on the contact distribution $\theta=0$.  It is known, by Chow's theorem, that this is a path metric space whose geodesics are the minimizers among horizontal paths.  Our proof uses the following lemma, due to Gromov and Bella\"iche on Carnot--Carath\'eodory spaces \cite{Bellaiche1996}:
\begin{lemma}
  Let $\mathcal U$ be a smooth connected contact manifold, with a contact one-form $\theta$, and a sub-Riemannian metric $g$ on the contact distribution $\theta=0$.  If $\mathcal U$ has a dilation, then it is the Heisenberg group: there exists a basis of smooth vector fields $P_i,Q^j$ of $\theta=0$ such that
  $$[P_i,Q^j]=\delta_i^jZ,\quad [P_i,P_j]=[Q^i,Q^j]=[Z,P_i]=[Z,Q^j]=0.$$
  with homogeneities
  $$\mathscr L_{D}P_i=-P_i,\quad\mathscr L_{D}Q^j=-Q^j,\quad\mathscr L_{D}Z=-2Z$$
  where ${D}$ is the infinitesimal generator of the dilation.  
\end{lemma}

We shall use this lemma to construct a Rosen metric on a spacetime admitting a dilation fixing a null geodesic.

Let $\gamma$ be the null geodesic fixed by the dilation group, and $u$ an affine parameter along $\gamma$.  Lift $\gamma$ to the corresponding trajectory $\widehat\gamma$ of $\nabla^\omega H$ in $\mathcal N$.  Also lift the infinitesimal generator ${D}$ of the dilation to the vector field $\widehat {D}=\mathscr L_{D} + 2E$ on $\mathcal N$.  We note the relations $\mathscr L_{\widehat {D}}\theta=2\theta$, $\widehat {D}(H) = 2H$.  Therefore $[\widehat {D},\nabla^\omega H]=0$, i.e., $\widehat {D}$ is a Jacobi field.  For each $u$, let $\mathcal U(u)$ be the unstable manifold of $\widehat {D}$ containing $\widehat\gamma(u)$.  A discussion is in \S\ref{UnstableManifolds}.

On $\mathcal U(u)$ define a metric by
$$\widehat g = -g_{ab}\theta^a\theta^b - g^{ab}\omega_a\omega_b $$
where $\omega_a=Dp_a$ is the connection one-form.  On $\theta=0$, $\widehat g$ is positive definite, because $\theta^a$ and $\omega_a$ are both spacelike.  Also $\mathscr L_{\widehat {D}}\widehat g=2\widehat g$.  So, by the lemma, $\mathcal U(u)$ is a Heisenberg group.  Fix $u$ and let $Z,P_i,Q^j$ be a Heisenberg basis of $\mathcal U(u)$.  Extend $Z,P_i,Q^i$ to vector fields on the trace of the family $\mathcal U(u)$, as $u$ varies, by Lie transport along $\nabla^\omega H$ (i.e., Jacobi transport).  Since Jacobi transport respects the relations
$$[P_i,Q^j]=\delta_i^jZ,\quad [P_i,P_j]=[Q^i,Q^j]=[Z,P_i]=[Z,Q^j]=0$$
$$\mathscr L_{\widehat {D}}P_i = -P_i,\quad \mathscr L_{\widehat {D}}Q^j = -Q^j,\quad \mathscr L_{\widehat {D}}Z = -2Z $$
the extended vector fields also satisfy these, and in addition commute with $\nabla^\omega H$.

The commuting vector fields $P_i,Z$ are linearly independent, and have $n$-dimensional integral manifolds in $\mathcal U(u)$ for each $u$, and in particular an integral manifold containing the preferred point $\widehat\gamma(u)$.  Because they are projections of linearly independent solutions to the Jacobi equation, for all except a discrete set of $u$, $\pi_*P_i,\pi_*Z$ are linearly independent.  Choosing $u$ so that they are linearly independent, we have constructed a basis of commuting vector fields of the tangent space of $M$: $\pi_*P_i,\pi_*Z,\pi_*(\nabla^\omega H)$.  Introduce coordinates such that
$$\partial_i = \pi_*P_i,\quad \partial_v = \pi_*Z,\quad \partial_u=\pi_*(\nabla^\omega H).$$
Then $g(\partial_i,\partial_j)$ depends only on $u$, because it is preserved by the dilation.  For the same reason, $g(\partial_v,\partial_i)=g(\partial_v,\partial_v)=0$.  Also $g(\partial_u, \partial_i)=0$ because $P_i$ are contact vector fields, and $g(\partial_u,\partial_u)=0$ because $\nabla^\omega H$ is null.

Finally, $g(\partial_u,\partial_v)=\theta(Z)$, which has weight zero under $\widehat{\mathcal D}$, and is therefore a function only of $u$.  But the scalar $\theta(Z)$ is also preserved along the flow of $\nabla^\omega H$, because $Z$ is a Jacobi field and $\theta$ is invariant, and is therefore a constant.  (Moreover, $\theta(Z) \ne 0$, because $Z\pmod{\theta=0}$ spans the brackets of the contact distribution with itself, which is non-zero by the contact hypothesis.) So we can without loss of generality take $\theta(Z)=1$.

  Now we are done: we have shown that the metric $g$ is of the form
$$g=2\,du\,dv - g_{ij}(u)dx^i\,dx^j$$
as required.

\subsubsection{Unstable manifolds}\label{UnstableManifolds}
A lemma from dynamics is the (un)stable manifold theorem for pseudohyperbolic diffeomorphisms \cite{Ruelle}:
\begin{lemma}
  Let $X$ be a Euclidean space $V\subset X$ a neighborhood of $0\in X$, $f:V\to X$ a $C^r$ diffeomorphism onto the open set $f(V)\subset X$, such that $f(0)=0$.  Consider the spectral decomposition of $X$ under $df(0)$.  Let $X^-$ be the invariant subspace of $X$ corresponding to the spectrum of $df(0)$ inside the closed unit disc, and $X^+$ be the invariant subspace corresponding to the eigenvalues $|\lambda| > \rho > 1$, where $\rho$ is a constant.  Then, for sufficiently small $R>0$, the {\em unstable manifold}
  $$\mathcal U^+ = \bigcap_{n=0}^\infty f^n\bigcap_{n=0}^m f^{-m}X(\rho^{m-n}R)$$
  is the graph of a $C^r$ map $g:X^+(R)\to X^-(R)$ where $g(0)=0$ and $dg(0)=0$.  The manifold $\mathcal U^+$ depends continuously on $f$ in the $C^r$ topology.
\end{lemma}
Moreover, the set $\mathcal U^+$ is mapped to itself under $f^{-1}$.  In our case, we take $f=\mathcal D_1$ to be the time one map of the dilation group $\mathcal D_t$.  The set $\mathcal U^+$, while a priori only invariant under $\mathcal D_{-1}$, is in fact invariant under the flow $\mathcal D_t$ for sufficiently small $t$ (see \cite{Shub} for details).  The manifold $\mathcal U^+=\mathcal U^+(\gamma_0)$ varies smoothly as the fixed point $\gamma_0$ varies along $\gamma$.  Therefore the $\mathcal U^+(\gamma_0)$ are locally the level sets of a pair of smooth functions $u,p_v:\mathcal N\to\mathbb R$, such that $du\wedge dp_v=0$.  We can then extend these functions by using invariance under $\mathcal D_t$ to all of $\mathcal N$, by setting $u=u\circ \mathcal D_t, p_v = p_v\circ \mathcal D_t$, and the level surfaces of $u,p_v$ are then minimal invariant manifolds of $\mathcal D_t$, which we here call the unstable manifolds.

Instead of appeal to this general theorem, we can construct the unstable manifolds more directly in this case.  The idea is that since $\widehat D$ commutes with the vector field $\nabla^\omega H$ on $T^*M$, the exponential map of the connection is equivariant with respect to the group $\widehat{\mathcal D}$, so the unstable manifolds are the exponentials of the negative part of the spectrum of $\widehat D$.  We must show that this can be made precise globally.

It is simplest to find the stable manifolds of the original dilation group $\mathcal D_t$ first.  The Lie derivative $\mathscr L_D$ defines an endomorphism $D_0$ of the tangent space $\gamma^{-1}TM$ because $\gamma$ is fixed by $\mathcal D$ and so $D=0$ on $\gamma$.  Because the exponential is equivariant under the dilation, the negative part of the spectrum $\ker(D_0+1)(D_0+2) \subset \gamma^{-1}TM$ is stable under $D_0$, and therefore its exponential is stable under $\mathcal D_t$.  But $D_0$ is strictly hyperbolic on the negative part of its spectrum, so the exponential map is complete on this subspace.

The argument for the unstable manifolds of the lifted group $\widehat{\mathcal D}_t$ follows the same idea.  Here we use the metric $-\widehat g = g_{ab}\theta^a\theta^b + g^{ab}\omega_a\omega_b$ as above, but now on all of the $2n+4$ dimensional manifold $T^*M$. The lifted dilation group $\widehat{\mathcal D}$ is a dilation of $\widehat g$, so the unstable manifold $\exp\ker(\widehat D_0+1)(\widehat D_0+2)$ for $\widehat g$ is a $2n+2$ dimensional manifold (we lose two dimensions now).  Here $\widehat D_0$ is the endomorphism induced on $\widehat\gamma^{-1}T(T^*M)$, where $\widehat\gamma$ is the trajectory of $\nabla^\omega H$ in the cotangent bundle, corresponding to the geodesic $\gamma$.  Now the unstable manifold for the restriction of $\widehat D$ to $\mathcal N$ is just the intersection with the null cone $\mathcal N$, clearly an invariant manifold for $\widehat D$.

\subsubsection{Example}

To fix some of the ideas of the proof, it is instructive to go codify its basic elements in the case when we know a priori that $g$ is a Rosen metric.  The canonical contact form is
$$\theta = p_udu + p_vdv + p_idx^i.$$
The generator of the scaling in the cotangent bundle is:
$$E=p_u\partial^{p_u} + p_v\partial^{p_v} + p_i\partial^i.$$
The dilation and the induced vector field on the cotangent bundle are
$${D} = 2v\partial_v + x^i\partial_i, \mathscr L_{D} = 2v\partial_v + x^i\partial_i - 2p_v\partial^{p_v} - p_i\partial^i.$$
The lift is therefore
$$\widehat {D} = \mathscr L_{D}+2E = 2v\partial_v + x^i\partial_i + 2p_u\partial^{p_u} + p_i\partial^i.$$
The Hamiltonian and its gradient are:
$$2H=2p_up_v - g^{ij}(u)p_ip_j,\quad \nabla^\omega H = p_u\partial_v + p_v\partial_u - g^{ij}(u)p_i\partial_j + \tfrac12\dot g^{ij}p_ip_j\partial^{p_u}.$$
The foliation by $\mathcal U$ is thus
$$\mathcal U(u_0) = \{u=u_0, p_v=1\} \subset \mathcal N$$
The Jacobi fields are:
$$P_i = \partial_i - p_i\partial_v, \quad Z=\partial_v,\qquad Q^j = (x^j-p_iH^{ij})\partial_v + H^{ij}\partial_i.$$
Finally, on the integral manifold of $P,Z$ containing $\gamma$, we have $p_i=0$, so that
$$\pi_*P_i = \partial_i,\quad \pi_* Z = \partial_v, \quad \pi_*(\nabla^\omega H) = \partial_u$$
and we recover the usual Rosen frame.

\section{Rosen universes}
This section contains a proof of Theorem \ref{RosenUniverse} characterizing Rosen universes.  Before proceeding to the proof, we investigate some examples.  Suppose a plane wave carries a global Brinkmann metric \eqref{Brinkmann1}
\begin{equation*}
G_\beta(p) = 2\,du\,dv + x^Tp(u)x\,du^2 - dx^Tdx.
\end{equation*}
A Rosen universe associated to the plane wave is $G_\rho(L^TL)=2\,du\,dv - dx^TL^TL\,dx$ where $L$ is a Lagrangian matrix for the Jacobi equation determined by the symmetric operator $p$.  The singular points (places where the tensor $G_\rho$ degenerates) are exactly the isolated points $\mathbb S\subset\mathbb U$ where the rank of $L$ drops.  We shall describe the structure of these singularities, which amounts to studying the degenerate points of the tensor $h=L^TL$.

The theorem can be made plausible by considering examples.  Consider the case of $n=1$, and $p$ is constant.  We consider the problem of placing a singular point at $u=0$.  The three cases are then essentially:
\begin{itemize}
\item $p=-1$: $\dot S + S^2-1=0$, $S=\tanh u$, $L=L_0\cosh u$.  [Here $S$ is globally regular, so there is no singular point.]
\item $p=0$: $\dot S + S^2=0$, $S=1/u$, $\dot L=SL$, $L=L_1u$.  
\item $p=1$: $\dot S + S^2+1=0$, $S=\cot u$, $L=L_1\sin u$.  
\end{itemize}
Only in the cases $p=0,p=1$ is there a singular point, and we have selected the unique $S$ putting it at $u=0$.  Note that in these cases, $S$ is $u^{-1} + O(1)$ and $L$ is $L_1u + O(u^3)$.  Thus in both singular cases ($p\ge 0$), we have $h=L_1^2u^2+O(u^4)$ ($L_1\not=0$), and in the nonsingular case ($p<0$), we have $h=O(1)$.

If $n>1$, and $p$ is constant, then the theorem is proven easily by eigendecomposing $p$ and applying the $n=1$ case to the factors.  In fact, if we were to assume that only $O(1)$ terms in $p$ contribute to the analysis of the singular point, we would already be done with one direction of the proof.  The actual proof must show that this plausible fact is indeed the case, by keeping track of the correction terms.

For the reverse direction, we again consider the case of $n=1$.  Let $h$ be a function with an isolated zero at $u=0$, assume that
$$K(u)= u^2/h(u)$$
extends smoothly at $u=0$ and $K'(0)=0$.  Then $K_0=K(0)\ne 0$, and $h=u^2/K_0 + O(u^4)$.  Going to $n>1$ is possible if we assume that we are in the case where $h(u)$ is diagonalizable (under a {\em constant} coordinate change).  Each singular point of $h(u)$, is either a singular point or regular point of one of the factors.  We already noted that the theorem is trivial at a regular point, and at a singular point it follows by the $n=1$ case.  The actual proof again involves keeping track of the corrections in the spectral decomposition of $K(u)$.

\subsection{Proof of Theorem \ref{RosenUniverse}}
Consider then a Rosen metric
$$G=2\,du\,dv - dx^Tg(u)dx$$
and suppose that the metric $g(u)$ degenerates at $u=0$.  Recall that the Rosen metric has the form
$$g(u) = L(u)^TL(u) $$
where $L(u)$ obeys
$$L''(u) + pL(u) =0.$$
and, away from the singular point $u=0$, there is a symmetric tensor $S(u)$ such that
$$L' = SL,$$
where
$$S' + S^2 + p =0.$$
By taking $u_0=0$, it suffices to prove the following:
\begin{theorem}
  A necessary and sufficient condition for the Rosen metric to come from a smooth Brinkmann metric in a neighborhood of the singular point $u=0$ is that the symmetric tensor
  $$K(u) = u^2g(u)^{-1},\quad u\ne 0,$$
  admits a smooth extension to $u=0$, whose derivative vanishes at $0$ $(K'(0)=0)$.
\end{theorem}
\begin{proof}
  {\em Necessity}:
  Assume that the stated condition holds.  We will then produce another Rosen metric in a neighborhood of $0$ that is smooth at this point.  Let $\partial_i$ be the basis of Jacobi fields giving the original Rosen metric.  We define a new  Lagrangian basis of Jacobi fields by
  $$X^i = x^i\partial_v + H^{ij}\partial_j $$
  where the symmetric tensor $H$ satisfies $\dot H(u) = g^{-1}(u)$, $u\ne 0$.  (It is easy to see that $X^i$ is a Killing vector, a more direct argument follows the proof.)  For such a new frame to produce a coordinate system that is not singular at $u=0$, we must have that $G(X^i,X^j)$ is smooth and defines a nonsingular quadratic form.  The $G(X^i,X^j)$ are the matrix coefficients of the tensor $H\dot H^{-1}H$.

Thus we want  to compute:
\[ G(u) = H(u)\dot{H}^{-1}(u) H(u) = H(u)g(u)H(u), \]
Here we have:
\[ g(u)K(u) = K(u) g(u) = u^2 I.\]
Note in particular that $g_0 K_0 = K_0 g_0 = 0$, where $g_0 = g(0)$ and $K(0) = K_0$.
So we have, for $M(u)$ smooth, such that $M(0) = I$, the identity:
\[K(u) = K_0 + u^2 M'(u).\]
Now, with
\begin{align*}
  H &= \int g^{-1}(u) du = \int u^{-2} K(u) du \\
      &= - K_0 u^{-1} + M(u),
\end{align*}
this gives:
\[ uH = - K_0 + uM(u).\]
Then we get:
\begin{align*}
  G(u) &= H\dot{H}^{-1} H = uH K^{-1} Hu \\
       &= (-K_0 + uM(u))K^{-1}(- K_0 + uM(u))\\
       &= K_0K^{-1}K_0  + M(u)gM(u) - K_0 K^{-1} uM(u) - uM(u) K^{-1} K_0\\
       &= (I - M'(u)g(u))K_0 + M(u)g(u)M(u) \\ &\qquad - (I - M'(u)g(u))uM(u) - uM(u) (I - M'(u)g(u)).
\end{align*}
This is plainly smooth near $u = 0$ with value at $u = 0$:
\[ G(0) = K_0 + g_0.\]
We would like $G(0)$ to be invertible.    
Now we have
\[ u^2 I = g(u)K(u) = g(u)(K_0 + u^2 M'(u)), \]
Work to order $u^2$, with $g(u) = g_0 + ug_1 + u^2 g_2 + O(u^3)$.
Then we get:
\begin{align*}
 g_0 K_0 &= K_0 g_0 = g_1 K_0 = K_0 g_1 = 0, \\
  I &= g_2 K_0 + g_0M'(0) \\
         &= K_0 g_2 + M'(0) g_0.
\end{align*}
If now 
\[ K_0 X = 0, \]
then we get:
\[ X = g_0M'(0) X, \]
So the kernel of $K_0$ is a subspace of the range of $g_0$.
Similarly, if  
\[ g_0 Y = 0, \]
then we get:
\[ Y = K_0g_2 Y.\]
So the kernel of $g_0$ is a subspace of the range of $K_0$.

Since we have also $g_0K_0 = K_0 g_0 = 0$, we get that the kernel of $g_0$ equals the range of $K_0$; vice versa, the kernel of $K_0$ equals the range of $g_0$. In particular we get that $G(0) = K_0 + g_0$ is invertible.   But $G(0)$ is the smooth limit of positive definite endomorphisms, so, being invertible, is itself positive definite.

  {\em Sufficiency}:
  Conversely, assume that a singular Rosen has a smooth Brinkmann continuation at $u=0$.
  Let $S$ be the solution of the Sachs equation corresponding to the Rosen form.  So we have
  $$\dot S + S^2 + p =0,\quad \dot L = SL,\quad g(u) = L^TL.$$
  Also, we have $\ddot L + p L=0$.  

  Consider a decomposition
  $$L = L_0 + u L_1 + O(u^2).$$
  We consider a block diagonalization of $L_0$, using a left and right projection operator $P,\widetilde P$, where $P$ is the orthogonal projection onto the image of $L_0$, $\widetilde P$ is the orthogonal projection whose kernel is the kernel of $L_0$:
  $$PL_0=L_0\widetilde P = L_0, \quad P^2=P, \widetilde P^2=\widetilde P$$
  $$QL_0 = L_0\widetilde Q=0,\qquad Q=I-P, \widetilde Q = I-\widetilde P.$$
  These orthogonal projection operators are thus those associated to the four subspaces of the matrix $L_0$:
  $$\operatorname{im} P = \operatorname{im} L_0, \quad \operatorname{im} \widetilde P = \operatorname{im} L_0^T, \quad \operatorname{im} Q = \operatorname{ker} L_0^T, \quad \operatorname{im} \widetilde Q = \operatorname{ker} L_0.$$
  
  In terms of the left and right bases then,
  $$L_0 = \begin{bmatrix}K_0&0\\0&0\end{bmatrix}.$$
  We have:
  $$L^{-1} = u^{-1}(QL_1\widetilde Q)^{-1} + O(1).$$
  Therefore $uL^{-1}$ is smooth, and hence so is $K=uL^{-1}L^{-T}u$.

  The condition on the derivative of $K$ at $u=0$ requires more control over the $O(1)$ term in the above.  Note that we have not yet used the Sachs equation in an essential way.   So suppose $\dot L = SL$.  [Lemma: $S$ has a decomposition $S=Au^{-1}+B(u)$ where $A,B$ are symmetric, $A$ is constant, and $B(u)$ is smooth.] Assuming the ansatz that $S=Au^{-1}+B(u)$, the Sachs equation now reads
  $$(A^2-A)u^{-2} + u^{-1}(AB+BA) + \dot B + B^2 + p = 0.$$
  So $A^2=A$ and $AB+BA=O(u)$.  Since $A$ is symmetric, it is an orthogonal projection.  The right hand side of $\dot LX = SLX$ (any $X\in\mathbb X^2$) remains continuous.  The ansatz then demands that $AL_0=0$.  Because of the rank, it then follows that $A=Q$ is the orthogonal projection whose kernel is the image of $L_0$, i.e., onto the kernel of $L_0^T$.  It then follows that $P=I-A$ is the orthogonal projection onto the image of $L_0$.

  The original equation $\dot L = SL$ now becomes
  $$L_1 + O(u) = u^{-1}QL_0 + BL_0 + QL_1 = BL_0 + QL_1\implies PL_1 = BL_0.$$
  In particular, $PL_1\widetilde Q=0$ because $L_0\widetilde Q=0$, and so the corresponding off-diagonal block of $L_1$ is zero.

  Putting $L_1=\begin{bmatrix}L_{11}&0\\ L_{21} & L_{22}\end{bmatrix}$ and $A=K_0+L_{11}u$,
  $$\begin{bmatrix}K_0+uL_{11} & 0\\ uL_{21} & uL_{22}\end{bmatrix}^{-1} = \begin{bmatrix}A^{-1} & 0\\ -L_{22}^{-1}L_{21}A^{-1}& u^{-1}L_{22}^{-1}\end{bmatrix},$$
  and $uL^{-1}L^{-T}u$ has no terms of order $u$.  (Note that because the projections $\widetilde P,\widetilde Q$ are complementary orthogonal projections, the block product of $L^{-1}L^{-T}$ is well-defined.)

\end{proof}

A lemma that was used but not proved is:
\begin{lemma}
Let $S$ be a solution of the Sachs equation and $L$ an associated solution of $\dot L=SL$.  Then $S=O(u^{-1})$ at any singular point.
\end{lemma}
\begin{proof}
  We shall use the fact that $S=\dot LL^{-1}$. Write $L=L_0+uL_1+O(u^2)$.  The $2n\times n$ matrix $\begin{bmatrix}L_0\\ L_1\end{bmatrix}$ is of full rank $n$.
Recall that if a block matrix is
  $$\begin{bmatrix}A & B\\ C & D\end{bmatrix} $$
  with $A$ and $\Delta = D-CA^{-1}B$ invertible, then
  $$\begin{bmatrix}A & B\\ C & D\end{bmatrix}^{-1} = \begin{bmatrix}A^{-1} + A^{-1}B\Delta^{-1}CA^{-1} & - A^{-1}B\Delta^{-1}\\ -\Delta^{-1}CA^{-1}& \Delta^{-1}\end{bmatrix}.$$
  Applied to the block decomposition of $L_0+uL_1$,
  $$\begin{bmatrix}K_0 + u L_{11} & uL_{12}\\ uL_{21} & uL_{22}\end{bmatrix}^{-1} = \begin{bmatrix} A^{-1} + O(u)& O(1)\\ O(1) & O(u^{-1})\end{bmatrix}.$$
  Hence
  $$L_1(L_0+uL_1)^{-1} = L_1\begin{bmatrix} A^{-1} + O(u)& O(1)\\ O(1) & 0\end{bmatrix} + O(u^{-1}) $$
  as required.
  
\end{proof}

We also provide the following lemma, not strictly required, but for the reader's convenience since it is perhaps not obvious that $x^i\partial_v + H^{ij}\partial_j$ indeed comprises a Lagrangian basis of Jacobi fields:
\begin{lemma}
  Let $L,S$ be a solution of the Sachs system: $\dot S+S^2+p=0$, $\dot L = SL$.  Let $H$ be a symmetric matrix such that $L\dot HL^T=I$.  Then, at a point where $L,H$ are both non-singular, $J=LH$ and $T=S + (LHL^T)^{-1}$ give a solution of the Sachs system.
\end{lemma}
\begin{proof}
  Note that
  $$\frac{d}{du}(LHL^T) = S(LHL^T) + (LHL^T)S + I.$$
  So
  \begin{align*}
    \frac{d}{du}\left(S + (LHL^T)^{-1}\right)
    &= -S^2-p - S(LHL^T)^{-1} - (LHL^T)^{-1}S -(LHL^T)^{-2} \\
    &= -T^2 - p.
  \end{align*}
  Also,
  $$\dot J = SLH + (LHL^T)^{-1}LH = TJ.$$
\end{proof}

\section{Heisenberg symmetries}

\subsection{Jacobi fields and symmetries of plane waves}
Suppose we start with a Brinkmann metric in the form \eqref{Brinkmann1}:
$$G = 2\,du\,dv + x^Tp(u)x\,du^2 - dx^Tdx.$$
The equation for an (abreast) Jacobi field is
$$\ddot J + pJ=0.$$
\begin{lemma}
  If $J$ is an abreast Jacobi field, then the vector field
  \begin{equation}\label{JacobiKilling}
    X_J=x^T\dot J\partial_v + J(\partial_x)
  \end{equation}
  is a Killing field.
\end{lemma}
\begin{proof}
  We have
  $$\mathscr L_{X_J}G = 2\,du\, d(x^T\dot J) + 2x^TpJ\,du^2 - 2dJ^Tdx=0.$$
\end{proof}

\begin{lemma}\label{BracketJacobi}
  IF $J,K$ are abreast Jacobi fields, and $X_J$ is defined by \eqref{JacobiKilling}, then
  $$[X_J,X_K] = \omega(J,K)\partial_v$$
  where $\omega(J,K)$ is the symplectic form $\omega(J,K) =J^T\dot K-K^T\dot J$.  
\end{lemma} 

This suggests the following:
\begin{definition}
  Let $(\mathbb J,\omega)$ be a symplectic vector space.  The associated Heisenberg algebra is the vector space $\mathbb R\oplus\mathbb J$, with the Lie bracket
  $$[ v\oplus x, v'\oplus x' ] = \omega(x,x')\oplus 0.$$
\end{definition}

\begin{theorem}
  Let $\mathbb J$ denote the space of solutions to $\ddot J + pJ=0$, and define $X_J$ by \eqref{JacobiKilling} for $J\in\mathbb R$.  Then the set of linear combinations $X_{k,J} = k\partial_v + X_J, k\in\mathbb R, J\in\mathbb J$, is a Heisenberg algebra of Killing vector fields of the plane wave $\mathbb M$.  Moreover, the vector field $X_{k,J}$ integrates to a global isometry of $\mathbb M$.
\end{theorem}
\begin{proof}
  The first part of the theorem is the content of Lemma \eqref{BracketJacobi}.

  We want to find the one-parameter local group of diffeomorphisms with
  \begin{align*}
    v' &= k + x^T\dot J\\
    x' &= J.
  \end{align*}
  We have $x=x_0 + tJ$, $v=v_0 + t(k + x_0^T\dot J) + \frac{t^2}{2}J^T\dot J$.  Taking $t=1$, the diffeomorphism is
  $$[u,v,x]\mapsto [u,v + k + \dot J^Tx + 2^{-1}J^T\dot J, x + J]$$
  which is obviously globally defined and preserves the metric everywhere.
\end{proof}

In the course of the proof, we identified a group of global isometries of a plane wave,
\begin{equation}\label{JacobiIsometry}
  T(k,J):[u,v,x] \mapsto [u, v+k+\dot J.x + 2^{-1}J.\dot J, x+J].
\end{equation}
\begin{definition}\label{HeisenbergGroupDef}
  The Heisenberg group on the symplectic vector space $(\mathbb J,\omega)$ is the group of pairs $(a,J)$, $a\in\mathbb R$, $J\in\mathbb J$, with the composition law
  $$(k',J')\circ (k,J) = (k + k' + 2^{-1}\omega(J,J'), J+J').$$
\end{definition}
One checks that the composition is associative, with identity $(0,0)$, and inverse $(k,J)^{-1}=(-k,-J)$.
\begin{theorem}\label{HeisenbergAction}
The action \eqref{JacobiIsometry} is an action of the Heisenberg group for the symplectic space $(\mathbb J,\omega)$, where $\omega(J,K)=J^T\dot K - K^T\dot J$.
\end{theorem}
  We only have to verify that $T(k',J')\circ T(k,J) = T( (k',J')\circ (k,J))$, which is a straightforward symbolic calculation.
  \subsection{Regular Rosen metrics}
Consider a Rosen universe \eqref{Rosen1}:
$$G = 2\,du\,dv - dx^Th(u)dx.$$
Suppose temporarily that the set $\mathbb S$ of singular points of $h$ is empty, so that $G$ is non-singular for all $u$.  One reason that Rosen metrics are nice is that their isometry group contains a $2n+1$ dimensional Heisenberg group that is easy to write down.  (There can be {\em additional} symmetries, whose discussion is deferred to the next article in this series.)  As usual, let $H(u)$ be a smooth symmetric $n\times n$ matrix, such that $dH = h^{-1}du$.  Then one sees by an easy calculation that metric is preserved by the vector fields
$$\partial_v, \partial_x, x^T\partial_v + H\,\partial_x.$$
These vector fields obey commutation relations for the $2n+1$ dimensional Heisenberg algebra:
$$[P,Q]=I\,Z $$
with $Z=\partial_v, P=\partial_x,$ and $Q=x^T\partial_v + H\,\partial_x$, and $I$ is the identity on $\mathbb X$.  The action integrates to give a global action of the $2n+1$ dimensional Heisenberg group:
\begin{equation}\label{Teqn}
  T(v_0,a,b):\left[u,v,x\right] \mapsto \left[u, v + v_0 + b^Tx + \tfrac12(b^T a + b^THb) , x + a + Hb\right]
\end{equation}
where the inverse is
$$T(v_0,a,b)^{-1} = T(-v_0,-a,-b).$$

Since $\mathbb X$ is a Euclidean space, $\mathbb X\oplus\mathbb X$ carries a natural symplectic form $\omega(x\oplus y,x'\oplus y') = y.x' - x.y'$.  Definition \ref{HeisenbergGroupDef} specializes to the following composition law on the set of tuples $[v,x,b], v\in\mathbb R, x,b\in\mathbb X$:
  \begin{equation}\label{OurHeisenberg}
    [v,a,b]\circ [v',a',b'] = [v + v' + 2^{-1}(b.a' - b'.a), a + a', b + b']
  \end{equation}
One checks that the composition is associative, with identity $[0,0,0]$, and inverse $[v,a,b]^{-1}=[-v,-a,-b]$.  Moreover, the center of the group is the set of $[v,0,0]$, for $v\in\mathbb R$.

We can therefore use the group to write down a formula for the transverse null cone through any point.\footnote{By a {\em transverse null cone}, we mean a null cone transverse to the wave fronts, i.e., the null cone after deleting its generator parallel to $\partial_v$.}
\begin{theorem}
  If the Rosen metric \eqref{Rosen1} is regular everywhere $(\mathbb S=\emptyset)$ then:
  \begin{itemize}
    \item $T(v_0,a,b)$ in \eqref{Teqn} defines an action of the Heisenberg group \eqref{OurHeisenberg}, transitive on the affine spaces $u=$constant, and transitive on the transverse null directions through any point; and
    \item the transverse null cone with vertex $(u_0,v_0,x_0)\in\mathbb M$ is the hypersurface:
      \begin{equation}\label{nullcones}
        0 = 2(v - v_0)  - (x-x_0)^T(H(u) - H(u_0))^{-1}(x - x_0).
      \end{equation}
    \end{itemize}
  \end{theorem}
  
  \begin{proof}

  The fact that it is an action is a restatement of Theorem \ref{HeisenbergAction}.  Now, because the affine transformations $T[v,a,0]$ are transitive on any wave front, the full group is transitive on each wave front as well.

  To prove that it is transitive on the transverse null directions, it is sufficient to prove it for the stabilizer of a single point, $(v,x)=(0,0)$.  This stabilizer is $T(0,-Hb,b)$.  Acting on the null covector $dv$, the isotropy action is
  $$dv \mapsto dv + b^Tdx$$
  which is therefore transitive on the null covectors normalized against $\partial_v$.
Putting these observations together, we conclude that the isometry group is transitive on the set of null directions transverse to any wave front.

  The only part that remains to be shown is the second bullet.      Note that, since $\dot H=h^{-1}$ is positive-definite, the matrix $H(u)-H(u_0)$ is invertible for all $u\not=u_0$.
Without loss of generality, we can suppose $u_0=0$ and $H(0)=0$. The stabilizer of the origin $u=v=0, x=0$, is $T(0,0,b)$.  The orbit of the central null geodesic $v=x=0$ under this stabilizer is
$$v - \tfrac12 b^THb=0,\quad x - Hb = 0.$$
That is,
$$2v = x^TH^{-1}x,$$
which is precisely \eqref{nullcones}.
\end{proof}

\subsection{Rosen universes}

This all works nicely in a coordinate system where the metric in the Rosen universe is globally regular, but the description of the symmetries breaks down if the set $\mathbb S\subset\mathbb U$ of points where the metric has a (removable) singularity is nonempty, because the primitive $H$ of $h^{-1}$ blows up at every point of $\mathbb S$.  Even worse, it is not obvious how to continue $H$ past such a singularity.  With the characterization of Theorem \ref{RosenUniverse}, we can obtain a canonical continuation of $H$ past the singularity.  Specifically,
\begin{lemma}\label{JumpCondition}
  Given an arbitrary initial condition $H(u_1)$ at a regular point $u_1\in\mathbb U-\mathbb S$, there exists a unique function $H$ that is smooth on $\mathbb U-\mathbb S$ such that:
  \begin{itemize}
  \item $\dot H(u) = h^{-1}(u)$ for all $u\in\mathbb U-\mathbb S$; and
  \item $(u-u_0)H$ has a smooth continuation across $u=u_0$ for all $u_0\in\mathbb S$.
  \end{itemize}
\end{lemma}
\begin{proof}
  Moving a singular point to $u=0$, we have $g(u)^{-1} = u^{-2}K_0 + K_1(u)$ on a punctured neighborhood $0<|u|<\epsilon$, where $K_0$ is constant and $K_1$ is smooth on $(-\epsilon,\epsilon)$.  Let $M_1(u)$ be an arbitrary primitive of $K_1$ on $(-\epsilon,\epsilon)$.  We then have
  $$H(u)=-u^{-1}K_0 + M_1(u) + A\,1_{(-\epsilon,0)}(u) + B\, 1_{(0,\epsilon)}(u)$$
  where $1$ denotes the indicator function of the interval, and $A,B$ are constant symmetric matrices.  The condition that $uH$ remain smooth at $u=0$ then guarantees that $A=B$, and so we have uniqueness of the continuation.
\end{proof}
The main theorem of this section is:
\begin{theorem}
  Suppose that $\dot H = h^{-1}$.  The following are equivalent:
  \begin{itemize}
  \item $(u-u_0)H$ has a smooth continuation across $u=u_0$ for all $u_0\in\mathbb U$; and
  \item each of the vector fields
  $$ Q^i = x^i\partial_v + H^{ij}\partial_j$$
  is equal on $\mathbb U-\mathbb S$ to a unique Jacobi field $Q^i$ defined on the entire plane wave (i.e., in Brinkmann coordinates which are regular on all of $\mathbb U$).
\end{itemize}
\end{theorem}
\begin{proof}
  Consider a singular point which we assume without loss of generality at $u=0$.  Let $L(u)$ be a Lagrangian matrix associated to the Rosen universe, with $\mathbb L=L(u)\pmod{GL(n)}\subset\mathbb J$ the associated Lagrangian subspace.  Choose a Lagrangian matrix $\bar L(u)$ of Jacobi fields such that $\bar L(0)$ is invertible, and let $\bar{\mathbb L}\subset\mathbb J$ be its associated Lagrangian subspace.  There is $\epsilon>0$ such that $L(u)$ and $\bar L(u)$ are both invertible on $(0,\epsilon)$.  For each $u\in\mathbb U$, define $\mathbb H(u)\subset\mathbb J$ to be the Lagrangian space of Jacobi fields vanishing at $u$.  Then $\mathbb H$ depends smoothly on $u$.  We shall give a concrete description of it momentarily.  Fix $u_0=\epsilon/2$.  By subtracting a constant, we can assume that $H(u_0)=0$.

  With these preliminaries in place, we relate them to the present problem.  On $(0,\epsilon)$, the vector fields $Q^i$ are a basis of the space $\mathbb H(u_0)$.  Therefore, they extend (uniquely) to Jacobi fields $\widetilde Q^i$ on $(-\epsilon,\epsilon)$.  We have to show that $Q^i=\widetilde Q^i$ on $0<|u|<\epsilon$.  

  We describe the space $\mathbb H(u)$ in two ways.

  \begin{itemize}
  \item There exists a symmetric operator $\bar H(u) : \mathbb H(u_0)\to\bar{\mathbb L}$, depending smoothly on $u$, such that $\mathbb H(u)$ is the set of vectors of the form $J + \bar H(u) J$ with $J\in\mathbb H(u_0)$, for each $u\in(-\epsilon,\epsilon)$.    
  \item There exists a symmetric operator $H(u) : \mathbb H(u_0)\to\mathbb L$, depending smoothly on $u$, such that $\mathbb H(u)$ is the set of vectors of the form $J + H(u) J$ with $J\in\mathbb H(u_0)$, for each $0<|u|<\epsilon$.    
  \end{itemize}
  Moreover, $\bar H'(u)$ and $H'(u)$ are positive-definite on their respective domains.  Because these describe the same Lagrangian subspace, they are related by a canonical transformation that sends $\mathbb L$ to $\bar{\mathbb L}$ and is the identity on $\mathbb H(u_0)$.

  Use a block decomposition of $\mathbb J$ relative to the pair of subspaces $\mathbb H(u_0), \bar{\mathbb L}$.  Then $\mathbb H(u)$ is the image of the $2n\times n$ matrix
  $$\begin{bmatrix}I\\ \bar H(u)\end{bmatrix}.$$
  A canonical transformation preserving $\mathbb H(u_0)$ is an upper triangular (symplectic) matrix 
  $$\begin{bmatrix}A&B\\ 0&A^{-T}\end{bmatrix} $$
  where $A^{-1}B$ is symmetric.  Now
  $$\begin{bmatrix}A&B\\ 0&A^{-T}\end{bmatrix}\begin{bmatrix}I\\ \bar H(u)\end{bmatrix} = \begin{bmatrix}A+B\bar H(u)\\ A^{-T}\bar H(u)\end{bmatrix} = \begin{bmatrix}I\\ A^{-T}\bar H(u)(A+B\bar H(u))^{-1}\end{bmatrix}(A+B\bar H(u)).$$
  Because $A$ is invertible and $\bar H'$ is positive definite, the singularities of $A+B\bar H(u)$ are isolated.  
  So, possibly after shrinking $\epsilon$, for $u\in (0,\epsilon)$,
  \begin{equation}\label{Htrans}
    H(u) = A^{-T}\bar H(u)(A+B\bar H(u))^{-1}.
  \end{equation}
  Define $\widetilde H(u)$ by \eqref{Htrans} on $0<|u|<\epsilon$.

  So we have $\widetilde Q^i=x^i\partial_v + \widetilde{H}^{ij}\partial_j$ on $0<|u|<\epsilon$.  It is sufficient to finally prove that $u\widetilde{H}$ has a smooth extension across $u=0$.  As in the proof of Lemma \ref{JumpCondition}, $\widetilde H(u) = -u^{-1}K_0 + M_1(u) + A\,1_{(-\epsilon,0)}(u) + B\, 1_{(0,\epsilon)}(u)$, where $M_1$ is smooth. Since $\bar H(u)$ is smooth and $\widetilde H(u)$ is a rational function of $\bar H$, we conclude that $A=B$.  
\end{proof}

\begin{corollary}
  Let $H$ be a primitive of $h^{-1}$ such that $(u-s)H(u)$ has a smooth continuation through $u=s$ at all $s\in\mathbb S$.
  The Heisenberg isometries are transitive on each wave front $u=u_0$ where $u_0\not\in\mathbb S$, and transitive on the set of null directions transverse to the wave front.
  The transverse null cone of the point $(u_0,v_0,x_0)$ is the hypersurface of $\mathbb M - \mathbb S$
  $$2(v-v_0) - (x-x_0)^T(H(u)-H(u_0))^{-1}(x-x_0) = 0,$$
  $$(u,v,x)\in \mathbb M - \mathbb S := (\mathbb U - \mathbb S)\times\mathbb R\times\mathbb X$$
\end{corollary}

The proof of the second statement is the same as the one already given in the case of regular Rosen metrics.

Because the set of transverse null directions through a regular wave front is $2n+1$ dimensional, this set is a principal homogeneous space for the $2n+1$ dimensional Heisenberg group.  Carrying the dilation fixing the central null geodesic along an isometry, we obtain:
\begin{corollary}\label{ManyDilations}
Any plane wave has at least a $(2n+1)$-dimensional family of compatible dilations.
\end{corollary}

\subsection{An example}
Consider the Brinkmann metric  with $n=1$ and $p=1$
$$2\,du\,dv + x^2\,du^2 - dx^2.$$
A Rosen universe that is singular at $u=0$ comes from the solution $J=\sin u$ to the Jacobi equation $J'' + J = 0$.  We then have the metric
$$g = 2\,du\,dv - \sin^2u\,dx^2.$$
In these coordinates, the isometries are generated by the vector fields $$\partial_v,\quad \partial_x, \quad x\partial_v - \cot u \partial_x,$$ which are clearly singular in these coordinates where $u=0$.





We make a change of coordinates to remove the singularity.  Let $\bar x \cos u=x\sin u$ and $\bar v=v + 2^{-1}x^2\tan u$ .
\begin{align*}
  d\bar x &= x\,\sec^2u\,du + \tan u\,dx\\
  d\bar v &= dv + x\tan u\,dx + 2^{-1}x^2\sec^2u\,du.
\end{align*}
Then
$$g = 2\,du\,d\bar v - \cos^2u d\bar x^2$$
is regular at $u=0$.

The vector field $\partial_x$ satisfies $\partial_x(d\bar x) = \tan u, \partial_x(d\bar v) = x\tan u$, so
$$\partial_x = \tan u\,\bar\partial_{\bar x} + x\tan u\,\bar{\partial}_{\bar v},$$
which vanishes at $u=0$.  We have $\partial_v(d\bar x)=0$, $\partial_v(d\bar v) = 1$, and so $\bar\partial_{\bar v} = \partial_v$. The remaining Heisenberg symmetry is
\begin{align*}
  Q&=x\,\partial_v - \cot u\,\partial_x \\
   &= x\partial_v - \cot u\left(\tan u\,\bar\partial_{\bar x} + x\tan u\,\bar{\partial}_{\bar v}\right)\\
  &= -\bar{\partial}_{\bar x}.
\end{align*}

Consider the null cone with vertex at $v=x=0$, $u_0=\pi/4$:
$$2v - (1-\cot u)^{-1}x^2 = 0.$$
Note that $v=\bar v - 2^{-1}\bar x^2\cot u, x = \cot u\,\bar x$.  So in the barred coordinates,
\begin{align*}
  0 &= 2\bar v - \cot u\,\bar x^2 - (1-\cot u)^{-1}\cot^2u\,\bar x^2\\
    &= 2\bar v - (\tan u-1)^{-1}\bar x^2,
\end{align*}
which is now regular across $u=0$.

        \section{Extended example}
        This section contains an extended example, illustrating the power of some of the formalism developed thus far.

Let $X$ be the space of $n\times n$ symmetric positive-definite matrices, and allow $GL(n,\mathbb R)$ to act on $X$ by $g\mapsto AgA^T$.  For any continuous curve $g$ in $X$, define a plane wave spacetime by
$$G(g) = 2\,du\,dv - dx^Tg(u)dx.$$

\begin{theorem}
  Suppose that $g(u)$ is an orbit of a one-parameter subgroup of $GL(n,\mathbb R)$ with infinitesimal generator $A\in\mathfrak{gl}(n,\mathbb R)$.  Then the Ricci curvature of $G(g)$ is
  $$\operatorname{Ric}(G(g)) = -\frac14\operatorname{tr}[(A+g_0A^Tg_0^{-1})^2]\,du\otimes du$$
  where $g_0=g(u_0)$ is any given starting point.
\end{theorem}
\begin{proof}
A Rosen metric of the form $G(g)$ can be described as follows.  There exists a factorization of the form $g(u)$ as $g(u)=L(u)^TL(u)$ where $L(u)$ is a vector of $n$ Jacobi fields.\footnote{Spanning a Lagrangian subspace in the space of all Jacobi fields abreast of the central null geodesic with the symplectic form $\omega(L_1,L_2) = \dot L_1^TL_2 - \dot L_2^TL_1$.}  Thus $L=[L_1\dots L_n]$ where
$$\ddot L + pL =0$$
and $p$ is the (symmetric) tidal curvature.  Introducing the symmetric matrix $S$ associated to the Sachs equation, we are thus led to the system
\begin{align*}
  \dot S + S^2 + p &= 0\\
  S + S^T &=0\\
  \dot L&=SL.
\end{align*}
Then
$$\ddot L + pL= (SL)^{\dot{}} + pL = \dot SL + S^2L + pL = 0.$$
Let $g=L^TL$.  Then
\begin{align*}
  \dot g &= 2 L^TSL\\
  \ddot g &= 4L^TS^2L + 2L^T\dot SL\\
         &= 2L^T(S^2-p)L\\
  \dot gg^{-1}\dot g &= 4L^TS^2L\\
  2^{-1}\dot g g^{-1}\dot g - \ddot g &= 2L^TpL.
\end{align*}
Summarizing, the Ricci curvature can be computed as
$$\operatorname{Ric}(G(g)) = \operatorname{tr}\left(4^{-1}g^{-1}\dot g g^{-1}\dot g - 2^{-1}g^{-1}\ddot g\right)\,du\otimes du.$$

Now assume that
$$g(u) = e^{Au}g_0e^{A^Tu}.$$
We have
  \begin{align*}
    \dot g &= Ag + gA^T\\
    \ddot g &= A^2g + gA^{2T} + 2AgA^T.
  \end{align*}
  Thus
  \begin{align*}
    (g^{-1}\dot g)^2 &= g^{-1}A^2g + A^{2T} + g^{-1}AgA^T + A^Tg^{-1}Ag\\
    g^{-1}\ddot g &= g^{-1}A^2g + A^{2T} + 2g^{-1}AgA^T\\
    (g^{-1}\dot g)^2 - 2g^{-1}\ddot g &= -g^{-1}A^2g - A^{2T}  - 3g^{-1}AgA^T + A^Tg^{-1}Ag.
  \end{align*}
  The trace is
  $$-\operatorname{tr}(A^2) -\operatorname{tr}(A^{2T}) - 2\operatorname{tr}(g^{-1}AgA^T).$$
  Now note that
  $$\operatorname{tr}(g^{-1}AgA^T)=\operatorname{tr}\left(e^{-A^Tu}g_0^{-1}e^{-Au}Ae^{Au}g_0e^{A^Tu}A^T\right)=\operatorname{tr}(Ag_0A^Tg_0^{-1}).$$
  Putting these together, and dividing by four, we have
  $$\operatorname{Ric}(G(g)) = -\frac14\operatorname{tr}\left[(A+g_0A^Tg_0^{-1})^2\right]\,du\otimes du $$
  as claimed.

\end{proof}

\begin{corollary}
  Consider a curve in $X$ given by
  $$g(u) = e^{Af(u)}g_0e^{A^Tf(u)} $$
  where $f(u)$ is a smooth function and $A\in\mathfrak{sl}(n,\mathbb R)$.  Then
  $$\operatorname{Ric}(G(g)) = -\frac14\operatorname{tr}[(A+g_0A^Tg_0^{-1})^2]|f'(u)|^2\,du\otimes du $$
\end{corollary}
\begin{proof}
  In the proof of the theorem, note
  \begin{align*}
    \ddot g &= (A^2g + gA^{2T} + 2AgA^T)|f'(u)|^2 + (Ag + gA^T)f''(u)
  \end{align*}
  Since $A$ is trace-free, the term with $f''(u)$ does not contribute to the Ricci curvature.
\end{proof}

\subsection{The $n=2$ case}
Let $X$ be the space of $3\times 3$ symmetric positive-definite matrices.  On $X$, put the quadratic form $g\mapsto \det g$.  The surface $\det g=1$, equipped with the pullback metric of the quadratic form,  is invariant under the action of $SL(2,\mathbb R)$, and is a model of the hyperbolic plane $\mathbb H$.    


In the special case of a geodesic in $\mathbb H$, the intersection of a plane with spacelike normal vector through the origin with the hyperboloid, the one-parameter group is conjugate to $\begin{bmatrix}e^u&0\\0&e^{-u}\end{bmatrix}$.  Similarly, in the special case of a horocycle, the curve is unipotent orbit in $\mathbb H$ (the intersection of a null plane, not through the origin).  Let $g_{ij}(u)$ be a geodesic or horocycle in $\mathbb H$ parameterized by arc length.  
\begin{corollary}
  In the case of a geodesic or horocycle, Ricci curvature of $G$ is
  $$\operatorname{Ric}(G(g)) = -\frac12 du\otimes du.$$
\end{corollary}
\begin{proof}
Let $g(u) = \begin{bmatrix}a(u)&b(u)\\b(u)&d(u)\end{bmatrix}$.  Any unit speed geodesic is conjugate by $SL(2,\mathbb R)$ to the basic geodesic $a(u) = e^u,d(u)=e^{-u},b(u)=0$ which is an orbit of the one-parameter group with infinitesimal generator $\begin{bmatrix}1/2&0\\0&-1/2\end{bmatrix}$, for which we have $-4^{-1}\operatorname{tr}( (A+A^T)^2 ) = -1/2$.  On the other hand, the Ricci curvature is invariant under $GL(2,\mathbb R)$, and we are done with this case.  The horocycle case is similar, as any such horocycle is conjugate to the basic one $a(u)=1,b(u)=u,c(u)=1+u^2$, whose plane wave also has Ricci curvature given by the Corollary.
\end{proof}

(In the case of a circle or hypercycle, the result of the previous section implies that the Ricci curvature is a negative constant, but the constant will depend on the particular circle.)

\subsection{Hyperbolic drawings}\label{HyperbolicDrawings}
\begin{definition}
A {\em hyperbolic drawing} shall mean a continuous function $f:\mathbb R\to\mathbb H$ from the reals to the hyperbolic plane, having the property that $f|_{[n,n+1]}$ is a segment of a geodesic or horocycle parameterized by hyperbolic arc length, for all $n\in\mathbb Z$.
\end{definition}
For any hyperbolic drawing, we show how to construct a conformal vacuum plane wave.  The conformal structure of the plane wave that we construct shall have the property that there is a smooth vacuum metric in the conformal class on any given real interval bounded from below.

Actually, our construction is more general than this.  Suppose that $\hat g_{ij}(u)$ be a piecewise smooth curve made of segments that are orbits of one-parameter subgroups of $SL(2,\mathbb R)$ (e.g., geodesic and horocycle segments), each segment of which is parameterized by an affine parameter on $[n,n+1]$, $n\in\mathbb Z$.  On each smooth segment, the associated plane wave $G(\hat g)$ has Ricci curvature $-r_n^2/2\,du\otimes du$ where $r_n$ is a positive constant.  We now reparameterize $\hat g^s_{ij}(u)=\hat g_{ij}(S(u))$ in a canonical way so that the function $u\mapsto g_{ij}(u)$ is $C^\infty$.

Define the bump function
$$\phi(u) = c\begin{cases}\exp(-u^{-2}-(1-u)^{-2})& 0<u<1\\ 0 & \text{otherwise}\end{cases} $$
where the constant $c$ is chosen so that $\int_{-\infty}^\infty \phi(u)\,du=1$.  Let
$$s(u) = \sum_{n=0}^\infty \phi(u+n),\qquad S(u) = \int_0^u s(t)\,dt.$$
Let
$$\hat g^s_{ij}(u) = \hat g_{ij}(S(u)).$$

\begin{corollary}
  $u\mapsto \hat g^s_{ij}(u)$ is smooth, and the plane wave $G(\hat g^s)$ has Ricci curvature
  $$\operatorname{Ric}(G(\hat g^s)) = -\frac{r_n^2}{2}|s(u)|^2\,du\otimes du.$$
\end{corollary}
\begin{proof}
  The first assertion is obvious.  For the second, again it is sufficient to check it for the basic geodesic.



\end{proof}

\begin{lemma}
  Let $G(g)$ be a plane-wave of dimension $n+2$.  Then $m^{-2}(u)G(g)$ is vacuum for any smooth positive function $m$ such that
  $$nm''(u) + p(u)m(u) = 0$$
  where $p(u)$ is the tidal Ricci curvature.    
\end{lemma}
\begin{proof}
  Let $\widetilde g = m^{-2}g$, $dU = du/m^2$.  Then
  $$\frac{d\widetilde g}{dU}\widetilde g^{-1}\frac{d\widetilde g}{dU} = 4\dot m^2 + m^2 \dot g g^{-1}\dot g - 4m\dot m\dot g $$
  $$\frac{d^2\widetilde g}{dU^2} = 2\dot m^2 g - 2\ddot m mg - 2\dot mm\dot g + m^2\ddot g.$$
  Hence
  $$\frac{d\widetilde g}{dU}\widetilde g^{-1}\frac{d\widetilde g}{dU}  - 2 \frac{d^2\widetilde g}{dU^2} = (m^2\dot gg^{-1}\dot g - 2m^2\ddot g) + 4\ddot mmg.$$
  Taking the trace gives $m(4pm + 4n\ddot m)$, which vanishes if and only if the metric
  $$m^{-2}G(g) = 2\,dU\,dv - m(u)^{-2}g(u)(dx,dx) $$
  is vacuum.
\end{proof}

\begin{theorem}\label{HyperbolicDrawingTheorem}
  Let $g(u)$ be a hyperbolic drawing for $u\in\mathbb U$ any interval bounded below.  Then there exists a reparameterization $g\circ s$ of $g$ such that $g\circ s$ is smooth, and a smooth positive function $m$ such that the plane wave
  $$G=2\,dU\,dv - m(u)^{-2}g_{ij}(s(u))\,dx^i\,dx^j$$
  is vacuum, where $dU=du/m(u)^2$.
\end{theorem}
\begin{proof}


  We are in $n+2=4$ dimensions, so with the $n=2$ case of the lemma, the local conformal factor giving a vacuum metric satisfies
  $$2m''(u) + p(u)m(u) = 0$$
  where $p(u)$ is the tidal Ricci curvature.  On a segment $[n,n+1]$, this is
  $$m''(u) = \frac{r_n^2}{4}\phi(u+n)^2m(u).$$
  We want to show that this differential equation admits a global solution.  Let $m(u) = e^{y(u)}$.  Then this is
  $$y''(u) + y'(u)^2 = \frac{r_n^2}{4}\phi(u+n)^2.$$
  With $z=y'$, this is
  $$z'(u) + z(u)^2 = \frac{r_n^2}{4}\phi(u+n)^2.$$
  We shall take the initial condition $z(0)=0$.  Because the right-hand side is bounded on compact subsets, it is sufficient to show that the solution $z(u)$ can never become negative.  Let $u_0=\inf\{u | z(u) < 0\}$.  If $u_0$ is finite, then clearly $u_0\in\mathbb N$, because if $z(u_0)=0$ and $u_0$ is not an integer, then $z'(u_0)>0$ and the solution is positive to the right of $u_0$ (contrary to the definition of $u_0$).  We shall assume that $u_0$ is finite and an integer (for contradiction) and can assume without loss of generality that $u_0=0$.

  So we have reduced the problem to the analysis of the differential equation
  $$z'(u) + z(u)^2 = \frac{r_0^2}{4}\phi(u)^2,\quad z(0) = 0 $$
  in an immediate neighborhood to the right of $u=0$.

  The unique solution in a neighborhood of $u=0$ satisfies the integral equation
  $$z(u) = \int_0^u \left(\frac{r_0^2}{4}\phi(t)^2 - z(t)^2\right)\,dt,$$
  and can be obtained by the Picard iteration $z_0=0$,
  $$z_{n+1}(u) = \int_0^u \left(\frac{r_0^2}{4}\phi(t)^2 - z_n(t)^2\right)\,dt.$$
  We claim that there exists an $\epsilon>0$ such that for all $u\in (0,\epsilon)$:
  $$0 \le z_n(u) \le \frac{r_0}{2}\phi(u).$$

  Moreover, it is clearly sufficient to establish just the upper bound.  By bounding the right-hand side of the Picard iteration, it is sufficient to show that
  \begin{equation}\label{phiestimate}
    \int_0^u \frac{r_0^2}{4}\phi(t)^2\,dt \le \frac{r_0}{2}\phi(u)
  \end{equation}
  for $u\in (0,\epsilon)$ (for some $\epsilon$).  Differentiating with respect to $u$ and using $\phi(0)=0$, this is equivalent to the estimate
  $$\frac{r_0^2}{4}\phi(u)^2 \le \frac{r_0}{2}\phi'(u).$$
  We have
  $$\phi'(u) = \phi(u)\left(2u^{-3} - 2(1-u)^{-3}\right).$$
  So the estimate we need becomes
  $$\frac{r_0}{2}\phi(u)\le 2u^{-3}-2(1-u)^{-3}.$$
  Since the right-hand side tends to $+\infty$ as $u\to 0^+$, while the left-hand side tends to zero, there exists an $\epsilon$ such that \eqref{phiestimate} holds for $u\in (0,\epsilon)$.
  $$0 \le z_n(u) \le \frac{r_0}{2}\phi(u).$$

  So, summarizing, we have shown that $z_{n+1}(u)\ge 0, u\in(0,\epsilon)$, so $u_0\ge\epsilon>0$, which is the desired contradiction and we are done.
\end{proof}

  \section*{Acknowledgements}
  We thank Redstart Roasters and Clemmie Murdock III for their hospitality during the preparation of this manuscript, Dingfang Yang for many questions during seminars, and Aristotelis Panagiotopoulos for useful discussions about plane waves.  We also thank Maciej Dunajski for pointing out a mistake in the first draft of this article.  The second author thanks Jaros\l aw Kopinski, Pawe\l\  Nurowski, and Katja Sagerschnig for the invitation to talk at Grieg Meets Chopin in Warsaw in 2023, and Maciej Dunajski for organizing the recent program on conformal geodesics in Banff where some of these ideas were presented.

\bibliographystyle{plain} 
\bibliography{planewaves} 

\end{document}